\documentclass[a4paper,twocolumn,11pt]{quantumarticle}
\pdfoutput=1
\usepackage[T1]{fontenc}
\usepackage{hyperref}

\usepackage{graphicx}
\usepackage{subcaption}
\usepackage[utf8]{inputenc}
\usepackage[english]{babel}
\usepackage{amsmath,mathtools,amssymb,amsthm}

\DeclarePairedDelimiterX{\ket}[1]{\lvert}{\rangle}{#1}
\DeclarePairedDelimiterX{\bra}[1]{\langle}{\rvert}{#1}
\newcommand{\ketbra}[2]{\left\vert #1 \right\rangle\! \left\langle #2 \right\vert}
\newcommand{\tr}{\operatorname{Tr}} % trace
\newcommand{\imeas}{\mathbf{IM}} % set of incoherent measurements
\newcommand{\io}{\mathbf{IO}} % set of incoherent operations
\newcommand{\dio}{\mathbf{DIO}} % set of detection-incoherent operations
\newcommand{\mio}{\mathbf{MIO}} % set of maximally incoherent operations
\newcommand{\sepo}{\mathbf{Sep}} % set of separable operators
\newcommand{\sepd}{\mathbf{SepD}} % set of separable density operators
\newcommand{\cloccm}{\overline{\mathbf{LOCCM}}} % closure of the set of LOCCM
\newcommand{\clocc}{\overline{\mathbf{LOCC}}} % closure of the set of LOCC

\newcommand{\cnot}{\mathrm{CNOT}}

\newcommand{\mcp}{measurement-cohering power}
\newcommand{\mep}{measurement-entangling power}
\newcommand{\scp}{state-cohering power}
\newcommand{\sep}{state-entangling power}

\theoremstyle{definition}
\newtheorem{definition}{Definition}
\theoremstyle{definition}
\newtheorem{theorem}{Theorem}
\theoremstyle{definition}
\newtheorem{lemma}{Lemma}

\newtheorem{proposition}[theorem]{Proposition}
\theoremstyle{remark}
\newtheorem{corollary}[theorem]{Corollary}

\begin{document}

\title{Maneuvering measurement-coherence into measurement-entanglement}

\author{Ho-Joon Kim}
\affiliation{Department of Physics, Hanyang University, Seoul, 04763, Republic of Korea}
\orcid{0000-0002-3634-5831}
\email{eneration@gmail.com}
\affiliation{Department of Mathematics and Research Institute for Basic Sciences, Kyung Hee University, Seoul 02447, Korea}
\author{Soojoon Lee}
\email{level@khu.ac.kr}
\orcid{0000-0003-2925-1017}
\affiliation{Department of Mathematics and Research Institute for Basic Sciences, Kyung Hee University, Seoul 02447, Korea}
\affiliation{School of Computational Sciences, Korea Institute for Advanced Study, Seoul 02455, Korea}
\maketitle

\begin{abstract}
    Quantum dynamics governs the transformation of static quantum resources, such as coherence and entanglement, in both quantum states and measurements. Prior studies have established that a quantum channel's \scp{} can be converted into the \sep{} without additional coherence. Here, we complete this coherence-to-entanglement paradigm by demonstrating that a channel's \mcp{} can likewise be converted into the \mep{}. This result reinforces, on the dynamical level, the intuition that entanglement emerges as a manifestation of coherence. To formalize this picture, we develop resource theories for measurement-cohering and measurement-entangling powers and characterize the structure of incoherent measurements to analyze measurement-coherence generation. Furthermore, we show that the \scp{} of a quantum channel is equivalent to the \mcp{} of its adjoint map, and a corresponding equivalence also exists between the \sep{} and the \mep{}.
\end{abstract}

\section{Introduction}
Quantum coherence is a fundamental manifestation of the superposition principle and lies at the heart of many quintessential quantum phenomena such as double-slit interference and Schr\"{o}dinger's cat. In recent decades, coherence has been formalized within the framework of quantum resource theories ~\cite{levi2014QuantitativeTheoryCoherent,baumgratz2014QuantifyingCoherence,winter2016OperationalResourceCoherence,chitambar2016ComparisonIncoherentOperations,streltsov2017QuantumCoherenceResource}, leading to a unified understanding of its operational roles in quantum sensing ~\cite{ramsey1950MolecularBeamResonance,czhang2019DemonstratingQuantumCoherence}, quantum computation~\cite{hillery2016CoherenceResourceDecision,lhshi2017CoherenceDepletionGrover,ahnefeld2022CoherenceResourceShors}, and thermodynamics~\cite{gour2022RoleQuantumCoherence,gour2023ErratumRoleQuantum}. In parallel, entanglement—another central quantum resource—has been extensively studied, and its deep connection to coherence has been clarified: coherence distributed across different subsystems naturally gives rise to entanglement. This connection has been explored both conceptually and operationally, yielding explicit protocols that convert state-coherence into state-entanglement under appropriate restrictions on quantum operations~\cite{mskim2002EntanglementBeamSplitter,asboth2005ComputableMeasureNonclassicality,vogel2014UnifiedQuantificationNonclassicality,killoran2016ConvertingNonclassicalityEntanglement,streltsov2015MeasuringCoherenceEntanglement,chitambar2016RelatingEntanglementCoherence,regula2018ConvertingMultilevelNonclassicality,shkim2021ConvertingCoherenceBased}.

While these developments have primarily focused on states, quantum measurements themselves also constitute carriers of non-classical resources~\cite{buscemi2005CleanPOVM,girolami2014ObservableMeasureCoherence,oszmaniec2017SimulatingPOVM,oszmaniec2019OperationalRelevanceResource,skrzypczyk2019RobustnessMeasurementDiscrimination,skrzypczyk2019AllSetsIncompatible,theurer2019QuantifyingOperationsApplication,bischof2019ResourceTheoryCoherence,lipka-bartosik2021OperationalSignificanceQuantum,guff2021ResourceTheoryQuantum,buscemi2024CompleteOperationalResource,linden2025HowUseArbitrary}. A quantum measurement can possess coherence with respect to a fixed basis, enabling the extraction of phase-sensitive information that cannot be accessed by incoherent measurements~\cite{theurer2019QuantifyingOperationsApplication,cimini2019MeasuringCoherenceQuantum,narasimhachar2025CoherentMeasurementCost}. Likewise, measurements can exhibit entanglement, allowing joint measurement strategies that outperform any LOCC implementable scheme~\cite{lipka-bartosik2021OperationalSignificanceQuantum,hshuang2021InformationTheoreticBoundsQuantum,schen2024OptimalTradeoffEntanglement}. Despite their importance, measurement-based resources remain less explored than their state-based counterparts. Recent work has begun to fill this gap by establishing resource theories for measurement-coherence and measurement-entanglement and by demonstrating that measurement-coherence can be converted into measurement-entanglement in static scenarios~\cite{hjkim2022RelationQuantumCoherence}.

\begin{figure}[htb]
	\centering
	\includegraphics[width=0.9\linewidth]{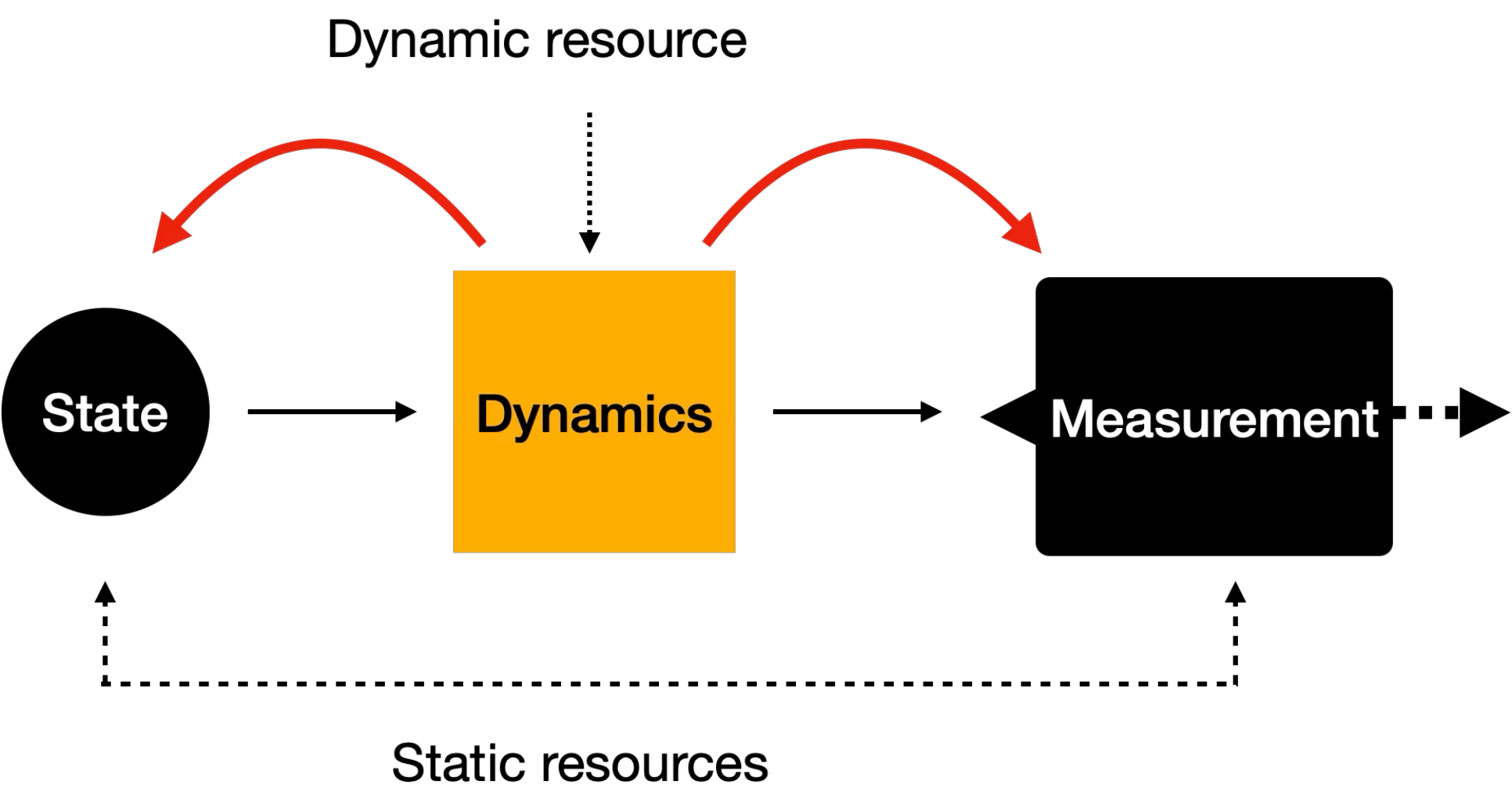}
	\caption{Quantum dynamics can modify static quantum resources in both quantum states and quantum measurements.}
	\label{fig: static-dynamic-resource}
\end{figure}
A complementary and equally fundamental aspect of quantum theory is quantum dynamics. Quantum channels do not merely transform states and measurements; they also serve as dynamical resources capable of generating~\cite{zanardi2000EntanglingPowerQuantum,zanardi2017CoherencegeneratingPowerQuantum,takahashi2022CreatingDestroyingCoherence}, destroying~\cite{horodecki2003EntanglementBreakingChannels,zwliu2017ResourceDestroyingMaps}, or preserving static quantum resources~\cite{cyhsieh2020ResourcePreservability} as depicted in Fig.~\ref{fig: static-dynamic-resource}. In addition, quantum channels themselves may function as dynamical resources, thereby facilitating the simulation of processes that do not possess intrinsic quantum resources~\cite{nielsen2003QuantumDynamicsPhysical,hjkim2021OneShotManipulationEntanglement,saxena2020DynamicalResourceTheory,yliu2020OperationalResourceTheory,yluo2025OneshotManipulationCoherence}. Viewed from this dynamical perspective, coherence and entanglement appear not only in static objects such as states or POVMs but also in the transformations applied to them. A recent result identified a relation between such dynamicals resources: the state-cohering power of a quantum channel can be fully converted into its state-entangling power, revealing a dynamical analogue of the coherence–entanglement correspondence known for states~\cite{theurer2020QuantifyingDynamicalCoherence}.

However, the dynamical landscape for measurement resources has remained largely unexplored. In particular, it is not known whether measurement-cohering power—the ability of a channel to generate measurement-coherence through pre-processing—can be operationally related to measurement-entangling power, or whether these two quantities obey a dynamical correspondence analogous to the one established for states. Establishing such a relationship is essential for completing the resource-theoretic picture linking coherence and entanglement across both static and dynamical settings.

In this work, we develop resource theories for quantum channels that quantify their ability to generate measurement-coherence and measurement-entanglement, and we show that these two forms of dynamical resources are fundamentally connected. We first demonstrate that the measurement-cohering power of any quantum channel can, without introducing any additional coherence, be fully converted into measurement-entangling power by means of a fixed construction involving only free operations and the final dephasing step. This establishes a genuine dynamical analogue of coherence-to-entanglement conversion on the measurement side.

A useful ingredient in our analysis is a structural characterization of the incoherent measurements: we show that every incoherent POVM can be reduced to a standard incoherent basis measurement followed solely by classical post-processing. This observation allows us to treat incoherent measurements as canonical objects and provides the foundation for analyzing how quantum channels generate measurement-resources.

Finally, we uncover a duality between state-based and measurement-based dynamical resource powers. For unital channels, the state-cohering power of the channel turns out to coincide with the measurement-cohering power of its adjoint map, and an analogous equivalence emerges between state-entangling and measurement-entangling powers. Taken together, these results complete the conceptual framework connecting coherence and entanglement across both static and dynamical settings, while highlighting a deeper symmetry between state- and measurement-based quantum resources and the unifying role played by quantum dynamics.

\section{Measurement-coherence and measurement-entanglement}
In the resource theory of quantum coherence~\cite{streltsov2017QuantumCoherenceResource}, quantum coherence of an operator is quantified with respect to a fixed orthonormal basis, which is referred to as an incoherent basis. The choice of incoherent basis is contingent upon the experimental or physical constraints imposed on the system under consideration. We set an incoherent basis of a system $A$ and denote it as $\mathcal{I}_{A} = \{\ket{i}_{A}: i=0, \dots, d-1\}$, with respect to which quantum coherence is considered. An operator $O_{A}$ that is diagonal in the incoherent basis is referred to as incoherent. Any incoherent operator is a fixed point under the dephasing channel $\Delta_{A}(O_{A})= \sum_{i=0}^{d-1} \bra{i}O_{A}\ket{i}_{A} \ketbra{i}{i}_{A} = O_{A}$. Similarly, a quantum state $\rho_{A}$ is incoherent if $\Delta_{A}(\rho_{A})=\rho_{A}$. Here, we consider constituent systems of dimension $d$ and assume that a quantum measurement for each system has $n$ outcomes, which can be described as a positive-operator-valued measure (POVM) given by $\mathcal{M}_{A}=\{M_{x} : \sum_{x=0}^{n-1} M_{x}=I_{A},M_{x}\ge 0\; \forall x=0, \ldots, n-1\}$, where $I_{A}$ is the identity operator. A quantum measurement $\mathcal{M}_{A}$ can be equivalently described by a quantum channel with classical outputs, as
\begin{equation}\label{eq: measurement as a channel}
	\mathcal{M}_{A}(\rho_{A}) = \sum_{x=0}^{n-1}\tr_{A}(M_{x}\rho_{A}) \ketbra{x}{x}_{A'},
\end{equation}
where $\rho_{A}$ is a quantum state and $A'$ denotes the output system~\cite{watrous2018TTQI}. The output of a quantum measurement is regarded as being classical implying that $\{\ketbra{x}{x}_{A'}:x=0,\dots, n-1\}$ is an incoherent basis for the system $A'$. A quantum measurement $\mathcal{M}_{A}$  is called an incoherent measurement if its POVM elements are incoherent, that is, $\Delta_{A}(M_{x})=M_{x}$ for all $x=0,\dots, n-1$. If a quantum measurement $\mathcal{M}_{A}$ is incoherent, the probability of an outcome is solely dependent on the incoherent information of input states, disregarding the coherent components~\cite{theurer2019QuantifyingOperationsApplication,kdwu2021ExperimentalProgressQuantum}.
We denote the set of incoherent measurements as $\imeas$. Measurement-coherence of a quantum measurement $\mathcal{M}_{A}$ is quantified by the relative entropy of measurement-coherence $C_{m}(\cdot)$, as defined in Ref.~\cite{hjkim2022RelationQuantumCoherence}:
\begin{equation}
	C_{m}(\mathcal{M}_{A}) = \min_{\mathcal{F}_{A}\in \imeas} D_{m}(\mathcal{M}_{A}\Vert \mathcal{F}_{A}),
\end{equation}
where the measurement relative entropy between two quantum measurements $D_{m}(\cdot\Vert \cdot)$ is proportional to the average relative entropy between POVM elements as
\begin{equation}
	D_{m}(\mathcal{M}_{A}\Vert \mathcal{N}_{A}) \coloneqq \dfrac{1}{d}\sum_{x=0}^{n-1} D(M_{x}\Vert N_{x}),
\end{equation}
with the quantum relative entropy $D(\cdot\Vert \cdot)$.
The relative entropy of measurement-coherence of a quantum measurement is found to be equal to the average relative entropy of coherence of each POVM element, up to a multiplicative constant~\cite{hjkim2022RelationQuantumCoherence}:
\begin{equation}\label{eq: measurement relative entropy of coherence reduced}
	C_{m}(\mathcal{M}_{A}) = \dfrac{1}{d}\sum_{x=0}^{n-1} \{S(\Delta_{A}(M_{x}))- S(M_{x})\},
\end{equation}
where $S(X_{A})=-\tr_{A} X_{A} \log X_{A}$ for $X_{A}\ge 0$ is the von Neumann entropy.

A quantum measurement $\mathcal{M}_{A}$ undergoes a transformation into another quantum measurement $\widetilde{\mathcal{M}}_{A}$ when a pre-processing channel and a classical post-processing channel are added~\cite{buscemi2005CleanPOVM,hjkim2022RelationQuantumCoherence}. This can be described as follows:
\begin{equation}
	\widetilde{\mathcal{M}}_{A}=\mathcal{S}_{A'}\circ \mathcal{M}_{A}\circ \mathcal{E}_{A},
\end{equation}
where a classical post-processing channel $\mathcal{S}_{A'}(\rho_{A'})=\sum_{x,x'=0}^{n-1} p(x'\vert x)\bra{x}\rho_{A'}\ket{x}_{A'} \ketbra{x'}{x'}_{A'}$ sends an outcome $x$ to another outcome $x'$ with a conditional probability distribution $p(x'\vert x)$ satisfying $\sum_{x'}p(x'\vert x)=1$ for each $x$.
At this point, we demonstrate that an incoherent measurement is, in essence, the incoherent basis measurement:
\begin{theorem}
	An incoherent measurement $\mathcal{M}_{A}\in \imeas$ can be implemented by the incoherent basis measurement $\mathcal{I}_{A}=\{\ketbra{i}{i}_{A}:i=0,\dots,d-1\}$ followed by a classical post-processing channel $\mathcal{S}_{A}$:
	\begin{equation}
		\mathcal{M}_{A} = \mathcal{S}_{A}\circ \mathcal{I}_{A}.
	\end{equation}
\end{theorem}
We defer proofs of the above and subsequent results to Appendices.
The set of quantum channels that do not generate any measurement-coherence is known to be detection-incoherent operations ($\dio$)~\cite{theurer2019QuantifyingOperationsApplication}. These serve as free operations for measurement-coherence. A quantum channel $\mathcal{E}_{A}$ is in $\dio$ if and only if it satisfies the equation $\Delta_{A}\circ \mathcal{E}_{A}=\Delta_{A}\circ \mathcal{E}_{A}\circ \Delta_{A}$. This characterization is analogous to the characterization of a quantum channel in the maximally incoherent operations ($\mio$), which is the largest set of quantum channels that send incoherent states to incoherent states.

In the resource theory of measurement-entanglement, we take the closure of the set of LOCC measurements ($\cloccm$) as our free resources, that can be implemented as LOCC operations with local measurements~\cite{watrous2018TTQI}. The closure of the set of LOCC channels ($\clocc$) corresponds to the free operations that preserve $\cloccm$. Measurement-entanglement of a bipartite quantum measurement $\mathcal{M}_{AB}$ is quantified by the relative entropy of measurement-entanglement of $\mathcal{M}_{AB}$ defined as follows\footnote{We choose $\cloccm$ as our free resource, which possesses a more direct operational meaning than the set of separable measurements, the free resource used in Ref.~\cite{hjkim2022RelationQuantumCoherence}.}:
\begin{equation}
	E_{m}(\mathcal{M}_{AB}) \coloneqq \min_{\mathcal{F}\in \cloccm(A:B)}D_{m}(\mathcal{M}_{AB}\Vert \mathcal{F}_{AB}).
\end{equation}
We remark that, despite the similarities between the resource theories of quantum coherence and entanglement, there is no entanglement-destroying channel that destroys entanglement while holding every separable state as its fixed point~\cite{gour2017QuantumResourceTheories}. This is in contrast to the case of quantum coherence, where the dephasing channel destroys static quantum coherence while holding every incoherent state as a fixed point of the channel.

\section{Measurement-resource powers of quantum dynamics}
Now we focus on the capacity of quantum dynamics to generate measurement-resources. Quantum channels are capable of generating measurement-coherence and measurement-entanglement by acting as a pre-processing channel prior to a quantum measurement. In contrast, classical post-processing after a quantum measurement is unable to generate any measurement-resources~\cite{hjkim2022RelationQuantumCoherence}.

As previously introduced, $\dio$ is the largest set of quantum channels under which incoherent measurements remain incoherent. In light of this, we define the \mcp{} of a quantum channel $\mathcal{E}_{A}$ as follows:
\begin{align}\label{eq: def. mcgp}
	C(\mathcal{E}_{A}) \coloneqq& \min_{\mathcal{F}_{A}\in \dio}\max_{\mathcal{M}_{A\widetilde{A}}\in \imeas}\nonumber\\
    &D_{m}(\mathcal{M}_{A\widetilde{A}}\circ \mathcal{E}_{A}\Vert \mathcal{M}_{A\widetilde{A}}\circ \mathcal{F}_{A}),
\end{align}
where the ancillary system $\widetilde{A}$ can be arbitrary. The definition in Eq.~\eqref{eq: def. mcgp} may appear formal, but it has a direct physical interpretation.
The optimization over all incoherent measurements is unnecessary; the maximizing measurement is always the canonical incoherent basis measurement $\mathcal{I}_A$.
Intuitively, the measurement-cohering power captures how much the quantum channel $\mathcal{E}_A$ ``twists'' the classical axes of the standard incoherent measurement $\mathcal{I}_A$ into a superposition, thereby generating coherence in the measurement process.
This leads to the following operational simplification:
\begin{theorem}\label{thm: mcgp canonical form}
	For a quantum channel $\mathcal{E}_{A}$, the \mcp{} of the quantum channel is given by
	\begin{equation}\label{eq: mcgp equals to mc from inc. meas.}
		C(\mathcal{E}_{A}) = C_{m}(\mathcal{I}_{A}\circ \mathcal{E}_{A}),
	\end{equation}
	where $\mathcal{I}_{A}=\{\ketbra{i}{i}_{A}:i=0,\dots, d-1\}$ is the incoherent-basis measurement.
\end{theorem}
Theorem~\ref{thm: mcgp canonical form} ensures that the generalized Fourier transform corresponds to a quantum channel that has the maximum \mcp{} as in the case for the \scp{}~\cite{saxena2020DynamicalResourceTheory} since the generalized Fourier transform turns the incoherent basis measurement into a measurement with maximum measurement-coherence.

Furthermore, Theorem~\ref{thm: mcgp canonical form} establishes a direct pathway for experimental verification of the \mcp{}.
Since the \mcp{} of a channel $\mathcal{E}$ is equivalent to the measurement-coherence of the effective measurement $\mathcal{I} \circ \mathcal{E}$, it opens the possibility of quantifying it using well-established techniques for characterizing static coherence.
For instance, quantitative witnesses for coherence~\cite{djzhang2018EstimatingCoherenceMeasures,tzhang2024QuantificationEntanglementCoherence}, although originally developed for quantum states, might be adapted to quantify measurement-coherence, leveraging the fact that POVM elements are positive operators akin to unnormalized quantum states.
This bridges our dynamical framework with existing experimental protocols for static resources.

Similarly, we can quantify the capability of a channel to generate measurement-entanglement.
By identifying the closure of LOCC channels ($\clocc$) as the set of free resources, we define the \mep{} of a quantum channel $\mathcal{E}_{AB}$ as:
\begin{multline}\label{eq: def. megp}
	E(\mathcal{E}_{AB}) \coloneqq \min_{\mathcal{F}\in \clocc(A:B)}\max_{\mathcal{M}\in \cloccm(A\widetilde{A}:B\widetilde{B})}\\ D_{m}(\mathcal{M}_{A\widetilde{A}B\widetilde{B}}\circ \mathcal{E}_{AB}\Vert \mathcal{M}_{A\widetilde{A}B\widetilde{B}}\circ \mathcal{F}_{AB}),
\end{multline}
where the ancillary systems $\widetilde{A}$ and $\widetilde{B}$ can be arbitrary. In the definition of the \mep{}, the ancillary systems $\widetilde{A}$ and $\widetilde{B}$ cannot be omitted, in contrast to Theorem~\ref{thm: mcgp canonical form} for the \mcp{}, because they play essential operational roles. Without these ancillary systems, the channel that attains the maximum \mep{} would be the generalized CNOT gate. However, once ancillary systems are allowed, the generalized SWAP channel achieves the maximum \mep{}, as it converts locally implemented Bell measurements into measurements with maximum measurement-entanglement, paralleling the situation for \sep{}~\cite{eisert2000OptimalLocalImplementation}.

The measurement-resource powers in Eq.~\eqref{eq: def. mcgp} and Eq.~\eqref{eq: def. megp} are legitimate resource monotones that satisfy the following properties~\cite{chitambar2019QuantumResourceTheories}: i) the \mcp{} $C(\mathcal{E}_{A})$ is non-negative and faithful, that is, $C(\mathcal{E}_{A})$ is zero if and only if $ \mathcal{E}_{A} $ is in $\dio$. ii) The \mcp{} is monotone under a $\dio$ $\mathcal{K}_{A}$ and a unital $ \dio$ $\mathcal{L}_{A}$ such that $C(\mathcal{K}_{A}\circ \mathcal{E}_{A}\circ \mathcal{L}_{A})\le C(\mathcal{E}_{A})$. iii) The \mcp{} is convex, that is, it holds that $C(p\mathcal{E}_{A}+(1-p) \mathcal{G}_{A})\le pC(\mathcal{E}_{A})+(1-p) C(\mathcal{G}_{A})$ for quantum channels $\mathcal{E}_{A}$ and $\mathcal{G}_{A}$, and any $0\le p \le 1$.

Similarly, the \mep{} $E(\mathcal{E}_{AB})$ also satisfies the three properties: i) $E(\mathcal{E}_{AB})\ge 0$, with equality holding if and only if $\mathcal{E}_{AB}\in \clocc(A\!:\! B)$. ii) The \mep{} $E(\mathcal{E}_{AB})$ is monotone under the composition by a $\clocc$ channel $\mathcal{K}_{AB}$ and a unital $\clocc$ channel $\mathcal{L}_{AB}$ such that $E(\mathcal{K}_{AB}\circ \mathcal{E}_{AB}\circ \mathcal{L}_{AB})\le E(\mathcal{E}_{AB})$. iii) The \mep{} $E$ is convex, that is, for quantum channels $\mathcal{E}_{AB}$ and $\mathcal{G}_{AB}$, and any $0\le p \le 1$, it holds that $E(p\mathcal{E}_{AB}+(1-p) \mathcal{G}_{AB})\le pE(\mathcal{E}_{AB})+(1-p) E(\mathcal{G}_{AB})$.

\section{Conversion of \mcp{} to \mep{}}
Having established the concepts of measurement-cohering and \mep{}s of quantum channels, we now present our main result: quantum coherence can be converted into entanglement at the level of dynamics generating measurement resources as illustrated in Fig.~\ref{fig: mcgp to megp conversion}. First, we demonstrate that the \mcp{} of a quantum channel $\mathcal{E}_{A}$ serves as an upper bound on the \mep{} of a composite channel constructed from the quantum channel $\mathcal{E}_{A}$, without any additional measurement-cohering power:
\begin{theorem}\label{thm: mcgp bounds megp}
	Let $\mathcal{E}_{A}$ be a quantum channel. Then
	\begin{equation}
		E(\Delta_{AB}\circ\mathcal{K}_{AB}\circ \mathcal{E}_{A}\circ \mathcal{L}_{AB})\le C(\mathcal{E}_{A})
	\end{equation}
	for any $\dio$ channel $\mathcal{K}_{AB}$ and any unital $\dio $ channel $\mathcal{L}_{AB}$.
\end{theorem}
\begin{figure}[th]
	\centering
	\includegraphics[width=0.9\linewidth]{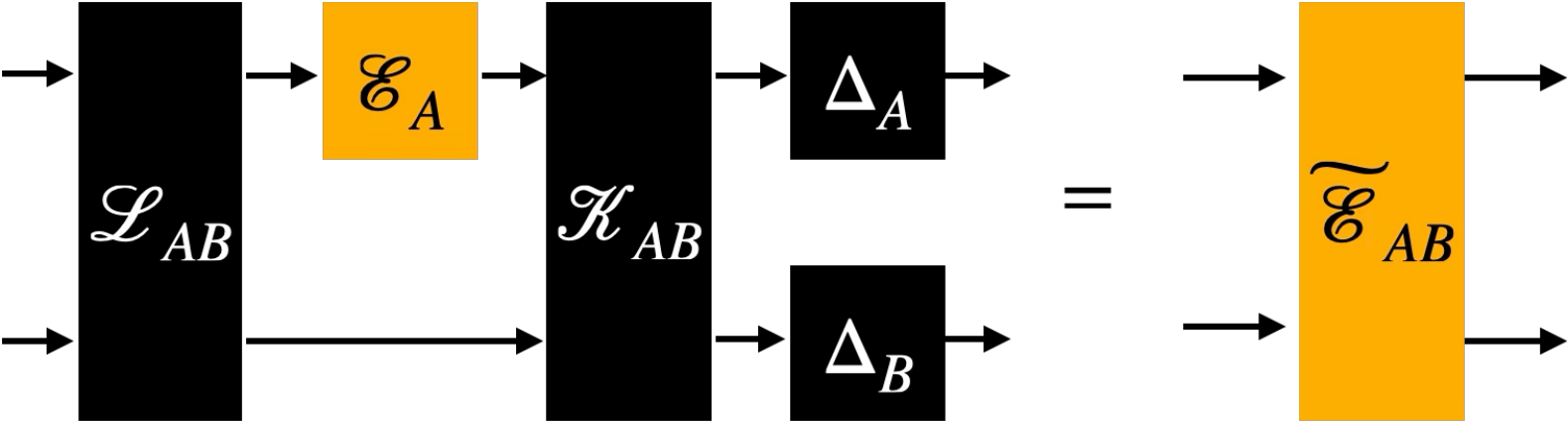}
	\caption{The \mcp{} of a quantum channel $\mathcal{E}_{A}$ can be converted to the \mep{} of a quantum channel $\widetilde{\mathcal{E}}_{AB}$ by composing it with a $\dio$ channel $\mathcal{K}_{AB}$, a unital $\dio$ channel $\mathcal{L}_{AB}$, and the dephasing channel $\Delta_{AB}$.}
	\label{fig: mcgp to megp conversion}
\end{figure}
Given that the channels surrounding the quantum channel $\mathcal{E}_{A}$ are $\dio$ channels, they do not contribute to the \mcp{} (see Fig.~\ref{fig: mcgp to megp conversion}). The dephasing channels in the post-processing serve to render the subsequent measurement incoherent and align the incoherent basis, thereby enabling the foregoing channels to convert the \mcp{} of the quantum channel $\mathcal{E}_{A}$ into the \mep{} in that incoherent basis. In contrast, when converting the \scp{} to the \sep{}, dephasing channels serve as a pre-processing step to ensure the input states are incoherent, as depicted in Table~\ref{fig: table of optimal conversions}~\cite{theurer2020QuantifyingDynamicalCoherence}. Consequently, Theorem~\ref{thm: mcgp bounds megp} implies that the \mcp{} of a quantum channel $\mathcal{E}_{A}$ is the sole source of the \mep{} of the composite quantum channel. Indeed, there exists a configuration that allows for complete conversion of the \mcp{} to the \mep{} although we restrict the pre-processing channel $\mathcal{L}_{AB}$ to be unital, as the following result shows (see Table~\ref{fig: table of optimal conversions}):
\begin{theorem}\label{thm: optimal mcgp conversion to megp}
	For a quantum channel $\mathcal{E}_{A}$, its \mcp{} can be fully converted to the \mep{} under a pre-processing by the adjoint channel of the generalized CNOT channel $\mathcal{U}_{\cnot}^{\dag}$ and a post-processing dephasing channel:
	\begin{equation}
		E(\Delta_{AB}\circ\mathcal{E}_{A} \circ \mathcal{U}_{\cnot}^{\dag})= C(\mathcal{E}_{A}),
	\end{equation}
	where $	U_{\cnot} = \sum_{i,j=0}^{d-1}\ketbra{i,i\oplus j}{ij}_{AB}$ with $\oplus$ denoting addition modulo $d$.
\end{theorem}
\begin{table*}[tb]
	\centering
	\renewcommand{\arraystretch}{1.5} % Adjust row spacing if needed
	\setlength{\tabcolsep}{12pt} % Adjust column spacing
	\begin{tabular}{|c|c||c|}
		\hline
		& \textbf{State-resource } & \textbf{Measurement-resource} \\
		\hline
		\textbf{Static} &
		\begin{minipage}{0.35\textwidth}
			\centering
			\vspace{5pt} % Space above figure
			\includegraphics[width=0.8\textwidth]{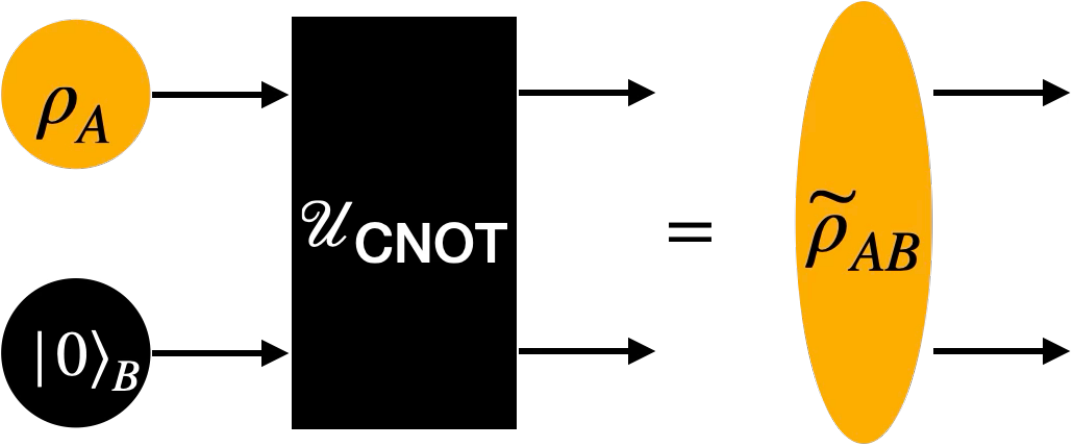}
			\vspace{5pt} % Space below figure
		\end{minipage} &
		\begin{minipage}{0.38\textwidth}
			\centering
			\vspace{5pt} 
			\includegraphics[width=0.9\textwidth]{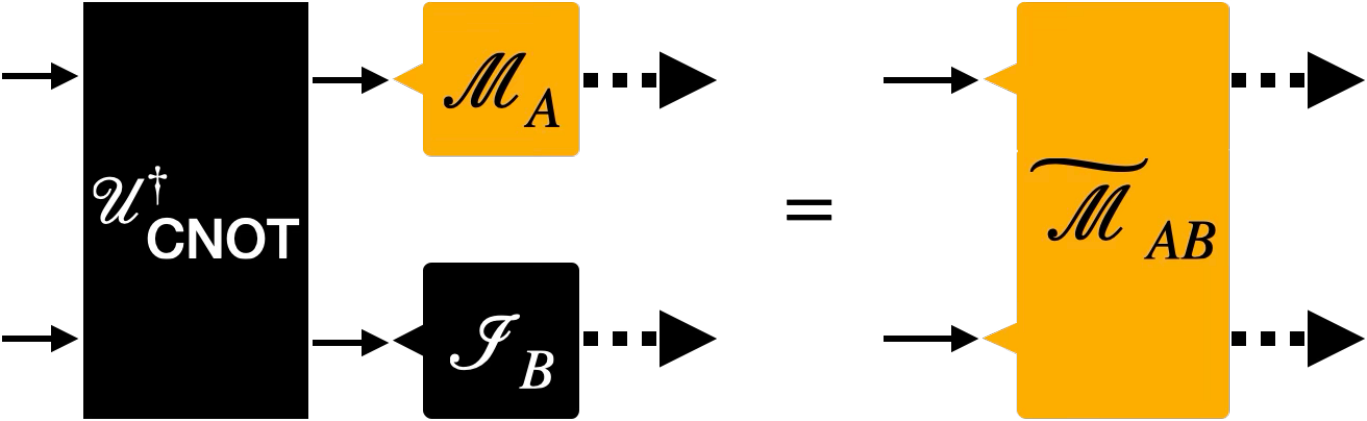}
			\vspace{5pt} 
		\end{minipage} \\
		\hline
		\textbf{Dynamic} &
		\begin{minipage}{0.36\textwidth}
			\centering
			\vspace{10pt} 
			\includegraphics[width=1.0\textwidth]{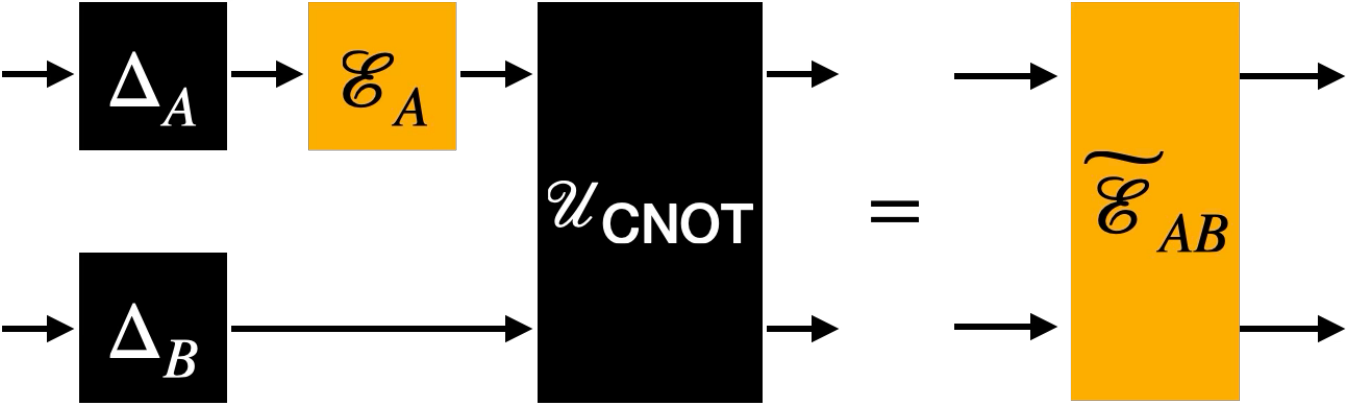}
			\vspace{1pt} 
		\end{minipage} &
		\begin{minipage}{0.38\textwidth}
			\centering
			\vspace{5pt} 
			\includegraphics[width=0.9\textwidth]{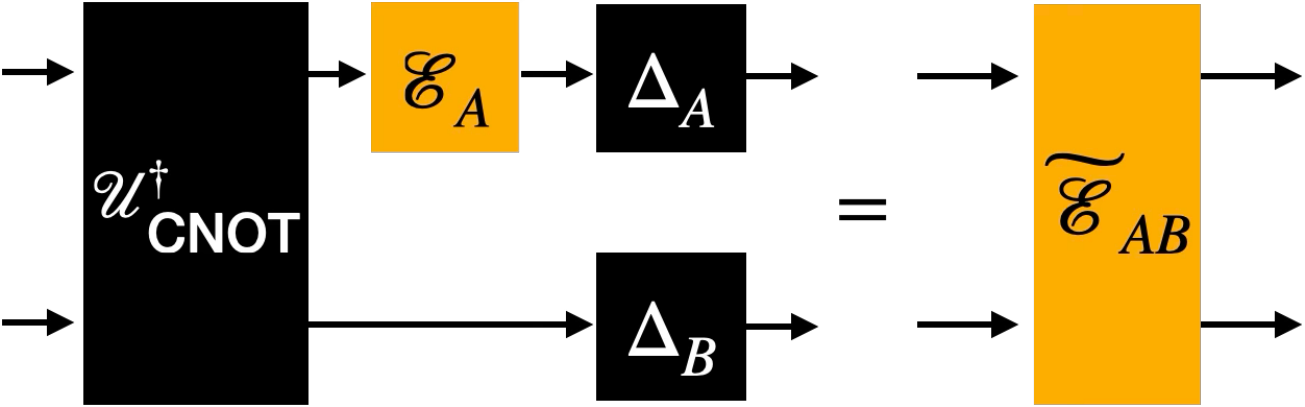}
			\vspace{5pt} 
		\end{minipage} \\
		\hline
	\end{tabular}
	\caption{Illustration of the optimal conversion of quantum coherence into entanglement at different levels of resource theory. The diagrams, arranged from the top left to the bottom left in a clockwise direction, depict: (i) the optimal conversion of quantum coherence of a quantum state $\rho_A$ into entanglement of a quantum state $\widetilde{\rho}_{AB}$~\cite{streltsov2015MeasuringCoherenceEntanglement}; (ii) the optimal conversion of measurement-coherence of a quantum measurement $\mathcal{M}_{A}$ into measurement-entanglement of a quantum measurement $\widetilde{\mathcal{M}}_{AB}$~\cite{hjkim2022RelationQuantumCoherence}; (iii) the optimal conversion of the \mcp{} of a quantum channel $\mathcal{E}_{A}$ into the \mep{} of a quantum channel $\widetilde{\mathcal{E}}_{AB}$ (this work); and (iv) the optimal conversion of the \scp{} of a quantum channel $\mathcal{E}_{A}$ into the \sep{} of a quantum channel $\widetilde{\mathcal{E}}_{AB}$~\cite{theurer2020QuantifyingDynamicalCoherence}.}
	\label{fig: table of optimal conversions}
\end{table*}
This result completes the picture of quantum coherence convertible to quantum entanglement at the dynamical level, specifically by addressing the previously unexplored role of quantum channels on the half of the static resources: measurement-resources. Measurement resources constitute half of the static resources in quantum information—complementing the widely studied state resources—yet have received comparatively less attention. This critical balance between measurement and state resources is fundamental, as state-coherence itself is inherently incomplete without coherent measurement to observe its effects~\cite{theurer2019QuantifyingOperationsApplication,narasimhachar2025CoherentMeasurementCost}. Similarly, measurement-entanglement is crucial in achieving exponential quantum advantage in quantum learning~\cite{hshuang2021InformationTheoreticBoundsQuantum,aharonov2022QuantumAlgorithmicMeasurement,schen2022ExponentialSeparationsLearning,schen2024OptimalTradeoffEntanglement}.

The convertibility of the \mcp{} to the \mep{} implies that a \mep{} monotone induces a \mcp{} monotone, as demonstrated by the following result:
\begin{theorem}
	A \mep{} monotone $\widetilde{E}$ induces a \mcp{} monotone $\widetilde{C}$ given by
	\begin{equation}
		\widetilde{C} (\mathcal{E}_{A}) = \max_{\mathcal{K}_{AB},\mathcal{L}_{AB}} \widetilde{E}(\Delta_{AB}\circ \mathcal{K}_{AB}\circ \mathcal{E}_{A}\circ \mathcal{L}_{AB}),
	\end{equation}
	where the maximization is over $\dio$ channels $\mathcal{K}_{AB}$ and unital $\dio $ channels $\mathcal{L}_{AB}$.
\end{theorem}

Our results reproduce those of static measurement-resources when considering measurement-resource powers for a quantum measurement described as a quantum channel, as in Eq.~\eqref{eq: measurement as a channel}. A quantum measurement $\mathcal{N}_{A}$ is a quantum channel whose output is incoherent, satisfying that $\mathcal{N}_{A}=\Delta_{A}\circ \mathcal{N}_{A}$. This, along with the observation that the dephasing channel corresponds to the incoherent basis measurement, leads to the following result:
\begin{theorem}
	For quantum measurements, the measurement-resource powers reduce to the relative entropies of measurement-resources: for a quantum measurement $\mathcal{N}_{A}$ and a bipartite quantum measurement $\mathcal{N}_{AB}$,
	\begin{equation}
		C(\mathcal{N}_{A})=C_{m}(\mathcal{N}_{A}),\quad E(\mathcal{N}_{AB})=E_{m}(\mathcal{N}_{AB}).
	\end{equation}
\end{theorem}

We now proceed to investigate the relation between the resource powers of quantum channels with respect to both static quantum resources.

\section{State-resource power versus measurement-resource power}
A quantum channel $\mathcal{E}_{A}$ can generate static quantum resources in  both quantum states and measurements. The \scp{} of $\mathcal{E}_{A}$, quantifying the maximum coherence the channel can induce in incoherent states, is given by
\begin{equation}
	C_{g}(\mathcal{E}_{A})\coloneqq \max_{\sigma_{A\widetilde{A}}\in \mathbf{I}} C_{R}(\mathcal{E}_{A}(\sigma_{A\widetilde{A}})),
\end{equation}
where $C_{R}(\rho_{A})\coloneqq \min_{\sigma_{A}\in \mathbf{I}} D(\rho_{A}\Vert \sigma_{A} )$ is the relative entropy of coherence~\cite{winter2016OperationalResourceCoherence}, and $\mathbf{I}$ denotes the set of incoherent states. The \mcp{} of a channel $\mathcal{E}_A$ is given in Eq.~\eqref{eq: def. mcgp}.

Are there relations between the two static resource powers of a quantum channel? In general, there is no relation between a channel's state-resource power and its measurement-resource power. Consider a qubit preparation channel $\mathcal{E}_{A}(\rho_{A}) = \tr_{A}(\rho_{A})\ketbra{+}{+}_{A}$ with $\ket{+}_{A}=(\ket{0}_{A}+\ket{1}_{A})/\sqrt{2}$. While $\mathcal{E}_{A}$ has the maximum state-cohering power ($C_g(\mathcal{E}_{A})=1$), it has zero measurement-cohering power ($C(\mathcal{E}_{A})=0$) because the effective measurement $\mathcal{I}_{A} \circ \mathcal{E}_{A}$ yields a trivial measurement whose POVM element is $I/2$. Another quantum channel $\mathcal{G}_{A}(\rho_{A})=\bra{+}\rho_{A}\ket{+}_{A}\ketbra{0}{0}_{A}+\bra{-}\rho_{A}\ket{-}_{A}\ketbra{1}{1}_{A}$ shows the opposite case, where it can generate no state-coherence as $C_{g}(\mathcal{G}_{A})=0$ while it can generate the maximum measurement-coherence $C(\mathcal{G}_{A})=1$. Meanwhile, the Hadamard gate possesses the maximum state-cohering power (converting incoherent basis states into maximally coherent states) and also the maximum measurement-cohering power (converting the incoherent measurement $\mathcal{I}$ into a measurement in the Fourier basis). However, for unital quantum channels, we establish a connection between the \scp{} and the \mcp{}:
\begin{theorem}
	For a unital quantum channel $\mathcal{E}_{A}$, the \scp{} $ C_{g}(\mathcal{E}_{A})$ and the \mcp{} $C(\mathcal{E}_{A}^{\dag})$ of its adjoint channel $\mathcal{E}_{A}^{\dag}$ satisfy the following:
	\begin{equation}
		\dfrac{1}{d}C_{g}(\mathcal{E}_{A}) \le C(\mathcal{E}_{A}^{\dag})\le C_{g}(\mathcal{E}_{A}).
	\end{equation}
\end{theorem}
This result implies that for a unital quantum channel, the \scp{} of the channel and the \mcp{} of its adjoint channel are equivalent in the sense that one is non-zero if and only if the other is non-zero.

For entanglement, we find a similar relation. The \sep{} of a quantum channel $\mathcal{E}_{AB}$\footnote{The term \sep{} has also been used to characterize the average state-entanglement generated from arbitrary product states as in Ref.~\cite{zanardi2000EntanglingPowerQuantum}. We follow the convention of Ref.~\cite{nielsen2003QuantumDynamicsPhysical} for consistency with the \scp{}.} is defined as
\begin{equation}
	E_{g}(\mathcal{E}_{AB}) \coloneqq \max_{\sigma\in \sepd(A\widetilde{A}:B\widetilde{B})}E_{R}(\mathcal{E}_{AB}(\sigma_{A\widetilde{A}B\widetilde{B}})),
\end{equation}
where $E_{R}(\rho_{AB})\coloneqq \min_{\sigma\in \sepd(A:B)} D(\rho_{AB}\Vert \sigma_{AB})$ is the relative entropy of entanglement~\cite{vedral1997}. The \sep{} of a unital quantum channel and the \mep{} of its adjoint channel are equivalent in the following sense:
\begin{theorem}
	For a unital quantum channel $\mathcal{E}_{AB}$, the \sep{} $ E_{g}(\mathcal{E}_{AB})$ and the \mep{} $E(\mathcal{E}_{AB}^{\dag})$ of its adjoint channel $\mathcal{E}_{AB}^{\dag}$ satisfy the following:
	\begin{equation}
		\dfrac{1}{d^4}E_{g}(\mathcal{E}_{AB}) \le E(\mathcal{E}_{AB}^{\dag})\le E_{g}(\mathcal{E}_{AB}).
	\end{equation}
\end{theorem}

\section{Conclusion}
We have developed resource theories for quantifying measurement-cohering and measurement-entangling powers of quantum channels, thereby extending the landscape of dynamical quantum resources to include measurement-based phenomena. Our results establish a dynamical correspondence between coherence and entanglement on the measurement side: the measurement-cohering power of any quantum channel can be fully converted into its measurement-entangling power using only free operations and a dephasing step. This conversion theorem completes the coherence--entanglement correspondence at the dynamical level and mirrors the known relationship between state-cohering and state-entangling powers, as illustrated in Table~\ref{fig: table of optimal conversions}.

A useful ingredient enabling these results is our structural characterization of incoherent measurements, showing that every incoherent POVM reduces to the incoherent basis measurement followed by classical post-processing. This structural insight allows measurement channels to be treated analogously to state preparations, permitting a unified treatment of their dynamical generation under quantum channels. Our framework also reproduces the static results of measurement-coherence and measurement-entanglement when a measurement is regarded as a channel with classical outputs.

Furthermore, we have shown that for unital channels, dynamical resource powers for states and measurements are tightly linked: the state-cohering power of a channel equals the measurement-cohering power of its adjoint map, and an analogous equivalence holds for state-entangling versus measurement-entangling powers. These dualities highlight a fundamental symmetry between state- and measurement-based quantum resources, suggesting that both arise from a single underlying structure governed by the dynamical action of quantum channels.

Despite these advances, several fundamental questions remain open. Although coherence-to-entanglement conversion is now understood in both static and dynamical regimes, it remains unclear whether all forms of entanglement can be obtained from coherence. Determining when, and to what extent, entanglement can be generated from coherence in a resource-theoretically consistent manner is an intriguing challenge. Moreover, the relationship between coherence and multipartite entanglement is still insufficiently explored~\cite{regula2018ConvertingMultilevelNonclassicality}, reflecting broader gaps in our understanding of multipartite entanglement itself.

Another important point is that static resource generation alone cannot capture the full complexity of quantum dynamics concerning resource utilization. Indeed, quantum dynamics itself can be viewed as a resource under appropriate restrictions on quantum superchannels~\cite{gour2019ComparisonQuantumChannels,gour2019HowQuantifyDynamical,gour2020DynamicalEntanglement,regula2021OneShotManipulationDynamical,xyuan2020OneshotDynamicalResource}. Recent works have shown that gates such as generalized Hadamard or SWAP can serve as fundamental dynamical resource units~\cite{saxena2020DynamicalResourceTheory,yluo2025OneshotManipulationCoherence,hjkim2021OneShotManipulationEntanglement,regula2021OneShotManipulationDynamical,xyuan2020OneshotDynamicalResource}, analogous to maximally coherent states or maximally entangled states in static resource theories. However, the connections among these dynamical quantum resources remain less understood compared to their static counterparts. Clarifying the relationship between dynamical coherence and dynamical entanglement resources may also prove valuable for quantum technologies, where the ability to manipulate channel-level resources underlies tasks such as fault-tolerant gate synthesis and quantum simulation. A unified understanding of these dynamical resources could help identify operations that are optimal or universal for emerging quantum devices.

Overall, our results place measurement-based dynamical resources on equal footing with their state-based counterparts and provide a unified theoretical foundation linking coherence and entanglement across static and dynamical levels. We expect that this framework will inform future advances in the broader development of resource-efficient quantum technologies, such as quantum learning and quantum simulation.

\begin{acknowledgments}
H-.J. K. was supported by the National Research Foundation of Korea(NRF) grant funded by the Korea government(MSIT)(No. RS-2023-NR119931). S. L. acknowledges support from the NRF grants funded by the MSIT  (No. RS-2022-NR068791 and No. RS-2024-00432214), Creation of the Quantum Information Science R\&D Ecosystem (No. RS-2023-NR068116) through the NRF funded by the MSIT, and the Institute of Information \& Communications Technology Planning \& Evaluation (IITP) grant funded by MSIT (No. RS-2025-02304540).
\end{acknowledgments}

\onecolumn
\appendix
\section*{Appendices}
\addcontentsline{toc}{section}{Appendices}

\vspace{1em}  % or try 2em or 12pt or 1cm, etc.

We assume that the systems under consideration have dimension $d$ and that a measurement produces $n$ outcomes for each system. The symbols $\io$, $\mio$, and $\dio$ are used to denote the set of incoherent operations, the set of maximally incoherent operations, and the set of detection-incoherent operations, respectively~\cite{streltsov2017QuantumCoherenceResource}. The set of incoherent measurements is written as $\imeas$. For the sake of clarity, the identity channel $\mathsf{id}_{A}$ in a system $A$ is often omitted for the sake of clarity.

\section{Resource theory of measurement-coherence and measurement-entanglement}
This section provides a concise overview of the essential elements of measurement-resource theories, primarily developed in Ref.~\cite{hjkim2022RelationQuantumCoherence}:
We quantify measurement-resources using the measurement relative entropy $D_{m}(\mathcal{M}_{A}\Vert \mathcal{N}_{A}) $ between measurements $\mathcal{M}_{A}=\{M_{x}\}_{x=0}^{n-1}$ and $\mathcal{N}_{A}=\{N_{x}\}_{x=0}^{n-1}$ defined by
\begin{equation}
	D_{m}(\mathcal{M}_{A}\Vert \mathcal{N}_{A}) = \dfrac{1}{d}\sum_{x=0}^{n-1} D(M_{x}\Vert N_{x}).
\end{equation}
Some properties of the measurement relative entropy are as follows:
\begin{lemma}
	Let $ \mathcal{M}_{A}, \mathcal{N}_{A}, \mathcal{K}_{A}, \mathcal{L}_{A}$ be quantum measurements with the output system $A'$, $ \mathcal{E}_{A} $ a unital quantum channel, and $ \mathcal{U}_{A} $ a unitary channel. Let $ \mathcal{S}_{A'} $ be a classical channel that sends $ \ket{x}_{A'} $ to $ \ket{y}_{A'} $ with a probability $ p(y\vert x) $ that satisfies $ \sum_{y} p(y\vert x)=1$ for all $ x $. Let $ 0\le p\le 1 $. The following holds:
	\begin{enumerate}
		\item $ D_{m}(\mathcal{M}_{A}\Vert \mathcal{N}_{A})\ge 0 $; the equality holds if and only if $ \mathcal{M}_{A}=\mathcal{N}_{A} $,
		\item $ D_{m}(\mathcal{M}_{A}\circ \mathcal{E}_{A}\Vert \mathcal{N}_{A}\circ \mathcal{E}_{A})\le D_{m}(\mathcal{M}_{A}\Vert \mathcal{N}_{A}),$
		\item $	D_{m}(\mathcal{M}_{A}\circ \mathcal{U}_{A}\Vert \mathcal{N}_{A}\circ \mathcal{U}_{A}) = D_{m}(\mathcal{M}_{A}\Vert \mathcal{N}_{A}), $
		\item $ D_{m}(\mathcal{S}_{A'}\circ \mathcal{M}_{A}\Vert \mathcal{S}_{A'}\circ \mathcal{N}_{A}) \le D_{m} (\mathcal{M}_{A}\Vert \mathcal{N}_{A})$,
		\item $	D_{m}(\mathcal{M}_{A}\otimes \mathcal{N}_{B}\Vert \mathcal{K}_{A}\otimes \mathcal{L}_{B}) = D_{m}(\mathcal{M}_{A}\Vert \mathcal{K}_{A}) +D_{m}(\mathcal{N}_{B}\Vert \mathcal{L}_{B}), $
		\item $	D_{m}(p\mathcal{M}_{A}+(1-p)\mathcal{N}_{A}\Vert p\mathcal{K}_{A}+(1-p)\mathcal{L}_{A}) \le pD_{m}(\mathcal{M}_{A}\Vert \mathcal{K}_{A})+(1-p)D_{m}(\mathcal{N}_{A}\Vert \mathcal{L}_{A}).$
	\end{enumerate}
\end{lemma}

In measurement-coherence resource theory, the set of incoherent measurements ($\imeas$) is the free resource, and the detection-incoherent operations ($\dio$) with classical post-processing channels correspond to the free operations that do not generate any measurement-coherence from incoherent measurements. The measurement-coherence of a measurement $\mathcal{M}_{A}=\{M_{x}\}$ is quantified by the measurement relative entropy of measurement-coherence defined as
\begin{align}
	C_{m}(\mathcal{M}_{A}) & \coloneqq \min_{\mathcal{F}_{A}\in \imeas} D_{m}(\mathcal{M}_{A}\Vert \mathcal{F}_{A}) \\
	& =\dfrac{1}{d}\sum_{x=0}^{n-1} \{S(\Delta(M_{x}))-S(M_{x})\},
\end{align}
where $S(X_{A})=-\tr_{A} X_{A}\log X_{A}$ for $X_{A}\ge 0$ is the von Neumann entropy.
In the measurement-entanglement resource theory introduced in Ref.~\cite{hjkim2022RelationQuantumCoherence}, the set of separable measurements and the set of separable channels were taken to be the free resources and free operations. Here, we adopt the closure of the set of LOCC measurement ($\cloccm$) and the closure of the set of LOCC channels ($\clocc$) with classical post-processing channels as the free resources and the free operations, respectively, due to their more direct operational meaning. The measurement-entanglement is quantified by the measurement relative entropy of measurement-entanglement defined as
\begin{equation}
	E_{m}(\mathcal{M}_{AB}) \coloneqq \min_{\mathcal{F}\in\cloccm(A:B)} D_{m}(\mathcal{M}_{AB}\Vert \mathcal{F}_{AB}).
\end{equation}

\begin{figure}[bh]
	\centering
	\includegraphics[width=0.5\linewidth]{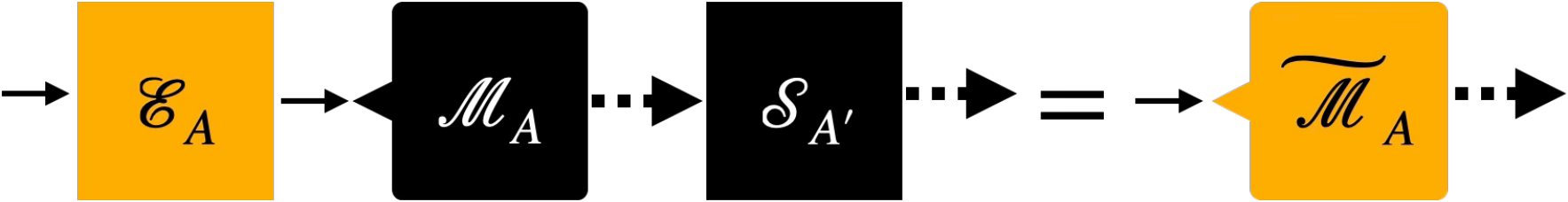}
	\caption{A quantum measurement $\mathcal{M}_{A}$ transforms to another quantum measurement $\widetilde{\mathcal{M}}_{A}$ under a pre-processing channel $\mathcal{E}_{A}$ and a classical post-processing channel $\mathcal{S}_{A'}$. The pre-processing channel can generate measurement-resource.}
	\label{fig: measurement transform}
\end{figure}
A quantum measurement can be transformed to another measurement by a pre-processing channel and a post-processing classical channel as depicted in Fig.~\ref{fig: measurement transform}. We demonstrate that an incoherent measurement is the incoherent basis measurement followed by a classical post-processing channel. It employs the observation that the dephasing channel can be seen as the incoherent basis measurement as follows:
\begin{proposition}\label{sup_thm: dephasing ch = incoh. basis meas.}
	The dephasing channel $\Delta_{A}$ with respect to the incoherent basis $\mathcal{I}_{A}=\{\ket{i}_{A}:i=0,\dots,d-1\}$ is equivalent to the quantum measurement in the incoherent basis $\mathcal{I}_{A}$ as a quantum channel, that is, $\mathcal{I}_{A}(\rho_{A})=\sum_{i=0}^{d-1} \tr_{A}(\ketbra{i}{i}_{A}\rho_{A})\ketbra{i}{i}_{A}$.
\end{proposition}

\begin{lemma}[Structure of an incoherent measurement]
	\label{sup_thm: structure of a IM}
	An incoherent measurement $\mathcal{M}_{A}=\{M_{x}\}_{x=0}^{n-1}\in \imeas$ is equivalent to an incoherent basis measurement $\mathcal{I}_{A}=\{\ketbra{i}{i}_{A}:i=0,\dots, d-1\}$ followed by a classical post-processing channel represented as $\mathcal{S}_{A}\circ \mathcal{I}_{A}$, where
	\begin{equation}
		\mathcal{S}_{A}(X_{A})= \sum_{x,i}p(x\vert i)\tr_{A} \left ( \ketbra{i}{i}_{A} X_{A}\right ) \ketbra{x}{x}_{A'}
	\end{equation}
	with a conditional probability distribution $p(x\vert i)=\bra{i}M_{x}\ket{i}_{A}\ge 0$ satisfying $\sum_{x}p(x\vert i)=1$ for all $i$ is a classical post-processing channel.
\end{lemma}
\begin{proof}
	We employ Proposition~\ref{sup_thm: dephasing ch = incoh. basis meas.} that the dephasing channel $\Delta_{A}$ is equivalent to the quantum measurement in the incoherent basis $\mathcal{I}_{A}$:
	\begin{align}
		\mathcal{M}_{A}(\mathcal{\rho_{A}}) & =\mathcal{M}_{A}\circ \Delta_{A}(\rho_{A}) \\
		& =\mathcal{M}_{A}\circ \mathcal{I}_{A}(\rho_{A}) \\
		&=\sum_{i} \tr_{A}(\ketbra{i}{i}_{A}\rho_{A})\mathcal{M}_{A}(\ketbra{i}{i}_{A})\\
		&=\sum_{i} \tr_{A}(\ketbra{i}{i}_{A}\rho_{A})\sum_{x}\tr_{A}(M_{x}\ketbra{i}{i}_{A})\ketbra{x}{x}_{A'}\\
		&=\sum_{x}\sum_{i}p(x\vert i)\tr_{A}(\ketbra{i}{i}_{A}\rho_{A}) \ketbra{x}{x}_{A'}\\
		&=\mathcal{S}_{A}\circ \mathcal{I}_{A}(\rho_{A}).
	\end{align}
	$\mathcal{M}_{A}$ is a POVM, hence we have $p(x\vert i)=\bra{i}M_{x}\ket{i}_{A}\ge 0$ and $\sum_{x}p(x\vert i)=1$ for all $i$.
\end{proof}

\begin{corollary}\label{sup_thm: structure of a bipartite IM}
	A bipartite incoherent measurement $\mathcal{M}_{AB}\in \imeas$ is equal to the incoherent basis measurement $\mathcal{I}_{AB}=\{\ketbra{ij}{ij}_{AB}:i,j=0,\dots, d-1\}$ followed by a classical post-processing channel $\mathcal{S}_{AB}\circ \mathcal{I}_{AB}$, where $\mathcal{S}_{AB}(X_{AB})= \sum_{x,y,i,j}p(x,y\vert i,j)\tr_{AB} ( \ketbra{ij}{ij}_{AB} X_{AB}) \ketbra{xy}{xy}_{A'B'}$ with a conditional probability distribution $p(x,y\vert i, j)=\bra{i,j}M_{xy}\ket{ij}_{AB}$ is a classical post-processing channel.
\end{corollary}
\begin{proof}
	A similar calculation as in Lemma~\ref{sup_thm: structure of a IM} shows the result.
\end{proof}
Note that Lemma~\ref{sup_thm: structure of a IM} and Corollary~\ref{sup_thm: structure of a bipartite IM} hold for any measurement $\mathcal{M}$ with an arbitrary number of measurement outcomes.

\section{Resource theory of measurement-resource powers}
\subsection{Measurement-cohering power}
The free static resource for measurement coherence is the set of incoherent measurements $\imeas$, and the largest set of free operations is $\dio$. We define the \mcp{} of a quantum channel as below:
\begin{definition}
	The \mcp{} of a quantum channel $\mathcal{E}_{A}$ is defined as
	\begin{equation}
		C(\mathcal{E}_{A}) \coloneqq \min_{\mathcal{F}_{A}\in \dio} \max_{\mathcal{M}_{A\widetilde{A}}\in \imeas} D_{m}(\mathcal{M}_{A\widetilde{A}}\circ \mathcal{E}_{A}\Vert \mathcal{M}_{A\widetilde{A}}\circ \mathcal{F}_{A}),
	\end{equation}
	where the maximization is over $\imeas$ with an arbitrary ancillary system $\widetilde{A}$.
\end{definition}
The following lemma demonstrates that the optimizations in the definition of the \mcp{} are unnecessary:
\begin{lemma}
	\label{sup_thm: meas-coh gen power}
	The \mcp{} of a quantum channel $\mathcal{E}_{A}$ is equal to the maximum measurement-coherence that can be generated via the quantum channel $\mathcal{E}_{A}$ from incoherent measurements. Furthermore, it is equal to the measurement relative entropy of measurement-coherence of the incoherent basis measurement $\mathcal{I}_{A}=\{\ketbra{i}{i}_{A}:i=0,\dots,d-1\}$ with the pre-processing channel $\mathcal{E}_{A}$:
	\begin{equation}
		C(\mathcal{E}_{A})= \max_{\mathcal{M}_{A}\in \imeas} C_{m} (\mathcal{M}_{A}\circ \mathcal{E}_{A}) =C_{m}(\mathcal{I}_{A}\circ \mathcal{E}_{A}).
	\end{equation}
\end{lemma}
For the proof of Lemma~\ref{sup_thm: meas-coh gen power}, we need the following result:
\begin{lemma}\label{sup_thm: no need for anc for IM}
	For quantum channels $\mathcal{E}_{A}$ and $\mathcal{F}_{A}$, it holds that
	\begin{equation}
		\max_{\mathcal{M}_{A\widetilde{A}}\in \imeas}D_{m}(\mathcal{M}_{A\widetilde{A}}\circ \mathcal{E}_{A}\Vert \mathcal{M}_{A\widetilde{A}}\circ \mathcal{F}_{A})	=\max_{\mathcal{M}_{A}\in \imeas}D_{m}(\mathcal{M}_{A}\circ \mathcal{E}_{A}\Vert \mathcal{M}_{A}\circ \mathcal{F}_{A}),
	\end{equation}
	where the ancillary system $\widetilde{A} $ can be arbitrary.
\end{lemma}
\begin{proof}
	The left-hand side (LHS) is equal to or greater than the right-hand side (RHS) due to the broader domain being considered for maximization. The opposite direction is demonstrated below by employing the structure of a bipartite incoherent measurement as stated in Corollary~\ref{sup_thm: structure of a bipartite IM}. Denoting the set of classical post-processing channels as $\mathbf{S}$, it follows that
	\begin{align}
		\max_{\mathcal{M}_{A\widetilde{A}}\in \imeas}D_{m}(\mathcal{M}_{A\widetilde{A}}\circ \mathcal{E}_{A}\Vert \mathcal{M}_{A\widetilde{A}}\circ \mathcal{F}_{A})&=\max_{\mathcal{S}_{A\widetilde{A}}\in\mathbf{S}}D_{m}(\mathcal{S}_{A\widetilde{A}}\circ \mathcal{I}_{A\widetilde{A}}\circ \mathcal{E}_{A}\Vert \mathcal{S}_{A\widetilde{A}}\circ \mathcal{I}_{A\widetilde{A}}\circ \mathcal{F}_{A}) \\
		& \le D_{m}( \mathcal{I}_{A\widetilde{A}}\circ \mathcal{E}_{A}\Vert \mathcal{I}_{A\widetilde{A}}\circ \mathcal{F}_{A}) \\
		&=D_{m}( \mathcal{I}_{A}\circ \mathcal{E}_{A}\Vert \mathcal{I}_{A}\circ \mathcal{F}_{A}) \\
		&\le \max_{\mathcal{M}_{A}\in \imeas}D_{m}(\mathcal{M}_{A}\circ \mathcal{E}_{A}\Vert \mathcal{M}_{A}\circ \mathcal{F}_{A}),
	\end{align}
	where the first inequality follows from the monotonicity of $D_{m}$ under classical post-processing channels.
\end{proof}
Here's the proof of Lemma~\ref{sup_thm: meas-coh gen power}:
\begin{proof}
	A quantum channel $\mathcal{F}_{A}$ is $\dio$ if $\mathcal{F}_{A}^{\dag}\circ \Delta_{A} = \Delta_{A}\circ \mathcal{F}_{A}^{\dag}\circ \Delta_{A}$. The \mcp{} of a quantum channel $\mathcal{E}_{A}$ is given by
	\begin{align}
		C(\mathcal{E}_{A}) &\coloneqq  \min_{\mathcal{F}_{A}\in \dio} \max_{\mathcal{M}_{A\widetilde{A}}\in \imeas} D_{m}(\mathcal{M}_{A\widetilde{A}}\circ \mathcal{E}_{A}\Vert \mathcal{M}_{A\widetilde{A}}\circ \mathcal{F}_{A})\\
		&=\min_{\mathcal{F}_{A}\in \dio} \max_{\mathcal{M}_{A}\in \imeas} D_{m}(\mathcal{M}_{A}\circ \mathcal{E}_{A}\Vert \mathcal{M}_{A}\circ \mathcal{F}_{A})\\
		&= \min_{\mathcal{F}_{A}\in \dio} \max_{\mathcal{M}_{A}\in \imeas}\dfrac{1}{d} \sum_{x} D( \mathcal{E}_{A}^{\dag}(M^{x}_{A})\Vert \mathcal{F}_{A}^{\dag}(M^{x}_{A})),
	\end{align}
	where we used Lemma~\ref{sup_thm: no need for anc for IM} for the second equality. For $\mathcal{M}_{A} \in \imeas$, we have that $\mathcal{F}_{A}^{\dag}(M^{x}_{A}) =\mathcal{F}_{A}^{\dag} \circ\Delta_{A}(M^{x}_{A})= \Delta_{A}\circ \mathcal{F}_{A}^{\dag}(M^{x}_{A}) $. Using $D(X\Vert \Delta(Y))=S(\Delta(X))-S(X)+D(\Delta(X)\Vert \Delta(Y))$, we find that
	\begin{align}
		C(\mathcal{E}_{A}) & =  \min_{\mathcal{F}_{A}\in \dio} \max_{\mathcal{M}_{A}\in \imeas} \dfrac{1}{d} \sum_{x} \{ S( \Delta_{A}\circ \mathcal{E}_{A}^{\dag}(M^{x}_{A})) - S(\mathcal{E}_{A}^{\dag}(M^{x}_{A}))\\
		&\qquad+ D(\Delta_{A}\circ \mathcal{E}_{A}^{\dag}(M^{x}_{A} )\Vert \Delta_{A}\circ \mathcal{F}_{A}^{\dag}(M^{x}_{A})) \}\\
		&=  \max_{\mathcal{M}_{A}\in \imeas} \dfrac{1}{d} \sum_{x}\left \{ S(\Delta_{A} \circ \mathcal{E}_{A}^{\dag}(M^{x}_{A})) -S(\mathcal{E}_{A}^{\dag}(M^{x}_{A})) \right \}\\
		&= \max_{\mathcal{M}_{A}\in \imeas} C_{m} (\mathcal{M}_{A}\circ \mathcal{E}_{A})\\
		&=\max_{\mathcal{S}_{A}\in\mathbf{S}}C_{m} (\mathcal{S}_{A}\circ\mathcal{I}_{A}\circ \mathcal{E}_{A})\\
		&=C_{m} (\mathcal{I}_{A}\circ \mathcal{E}_{A}),
	\end{align}
	where we choose $\mathcal{F}_{A}=\mathcal{E}_{A}\circ \Delta_{A}\in \dio$ in the second equality regarding the non-negativity of the relative entropy. The fourth equality follows from Lemma~\ref{sup_thm: structure of a IM}, and the final equality holds from the monotonicity of $C_{m}$ under classical post-processing channels.
\end{proof}

\begin{theorem}
	The \mcp{} $C(\mathcal{E}_{A})$ of a quantum channel $\mathcal{E}_{A}$ is a resource monotone satisfying
	\begin{enumerate}
		\item Non-negativity and faithfulness: $C(\mathcal{E}_{A})\ge 0$. $C(\mathcal{E}_{A}) = 0$ if and only if $\mathcal{E}_{A}\in \dio$.
		\item Monotonicity: For a detection-incoherent channel $\mathcal{K}_{A}$ and a unital detection-incoherent channel $\mathcal{L}_{A}$, it holds that
		\begin{equation}
			C(\mathcal{K}_{A}\circ \mathcal{E}_{A}\circ \mathcal{L}_{A}) \le C(\mathcal{E}_{A}).
		\end{equation}
		\item Convexity: For quantum channels $\mathcal{E}_{A}$ and $\mathcal{G}_{A}$, it holds that $C(p \mathcal{E}_{A}+(1-p) \mathcal{G}_{A})\le pC(\mathcal{E}_{A})+(1-p)C(\mathcal{G}_{A})$ for $0\le p \le 1$.
	\end{enumerate}
\end{theorem}
\begin{proof}
	\begin{enumerate}
		\item Non-negativity and faithfulness can be shown as follows: $C(\mathcal{E}_{A})\ge 0$ due to the non-negativity of $D_{m}$. $C(\mathcal{E}_{A})=0$ for $\mathcal{E}_{A}\in \dio$ from the definition. If $\mathcal{E}_{A}\notin \dio$, there exists $\mathcal{M}_{A\widetilde{A}}\in \imeas$ such that $\mathcal{M}_{A\widetilde{A}}\circ \mathcal{E}_{A}\notin \imeas$ implying $C(\mathcal{E}_{A})>0$ since $\dio$ is the largest set that preserves $\imeas$. Therefore, it holds that $C(\mathcal{E}_{A}) = 0$ if and only if $\mathcal{E}_{A}\in \dio$.
		\item For a detection-incoherent channel $\mathcal{K}_{A}$ and a unital detection-incoherent channel $\mathcal{L}_{A}$, it follows that
		\begin{align}
			C(\mathcal{K}_{A}\circ \mathcal{E}_{A}\circ \mathcal{L}_{A}) & = \max_{\mathcal{M}_{A}\in \imeas} C_{m} (\mathcal{M}_{A}\circ \mathcal{K}_{A}\circ \mathcal{E}_{A}\circ \mathcal{L}_{A})\\
			& \le \max_{\mathcal{M}_{A}\in \imeas} C_{m} (\mathcal{M}_{A}\circ \mathcal{E}_{A}\circ \mathcal{L}_{A})\\
			& \le \max_{\mathcal{M}_{A}\in \imeas} C_{m} (\mathcal{M}_{A}\circ \mathcal{E}_{A})\\
			&= C(\mathcal{E}_{A}),
		\end{align}
		where the first inequality follows due to $\mathcal{M}_{A}\circ \mathcal{K}_{A}\in \imeas$, and the second inequality holds from the monotonicity of $C_{m}$ under unital detection-incoherent channels. The last equality is given by Lemma~\ref{sup_thm: meas-coh gen power}.
		\item Convexity follows from the convexity of $C_{m}$ and Lemma~\ref{sup_thm: meas-coh gen power}:
		We first show that $C_{m}$ is convex: for a set of quantum measurements $\{\mathcal{M}_{i}\}$, let $C_{m}(\mathcal{M}_{i}) = D_{m}(\mathcal{M}_{i}\Vert \mathcal{F}_{i}^{\ast})$ for an incoherent measurement $\mathcal{F}_{i}^{\ast}$. Then, we have that, for a probability vector $\vec{p}$ satisfying $ \sum_{i} p_{i}=1$, it follows that
		\begin{align}
			C_{m} \left (\sum_{i}p_{i}\mathcal{M}_{i}\right )& =\min_{\mathcal{F}\in\imeas} D_{m}\left (\sum_{i}p_{i}\mathcal{M}_{i}\Vert \mathcal{F}\right ) \\
			& \le D_{m}\left (\sum_{i}p_{i}\mathcal{M}_{i}\Vert \sum_{i}p_{i}\mathcal{F}_{i}^{\ast}\right )\\
			&\le \sum_{i}p_{i} D_{m}\left (\mathcal{M}_{i}\Vert \mathcal{F}_{i}^{\ast}\right )\\
			&=\sum_{i} p_{i} C_{m}(\mathcal{M}_{i}),
		\end{align}
		where the second inequality follows from the joint convexity of the measurement relative entropy $D_{m}$. Now Lemma~\ref{sup_thm: meas-coh gen power} and convexity of $C_{m}$ imply that
		\begin{align}
			C(p \mathcal{E}_{A}+(1-p) \mathcal{G}_{A}) & = C_{m}(\mathcal{I}_{A}\circ (p \mathcal{E}_{A}+(1-p) \mathcal{G}_{A}))\\
			&\le p C_{m}(\mathcal{I}_{A}\circ \mathcal{E}_{A}) +(1-p) C_{m}(\mathcal{I}_{A}\circ \mathcal{G}_{A})\\
			&\le p C(\mathcal{E}_{A}) +(1-p) C(\mathcal{G}_{A}).
		\end{align}
	\end{enumerate}
\end{proof}

\subsection{Measurement-entangling power}
For measurement-entanglement, the free static resource is the closure of the set of LOCC measurements ($\cloccm$). Taking free dynamic resource as the closure of the set of LOCC channels ($\clocc$), we define the \mep{} of a quantum channel as follows:
\begin{definition}
	The \mep{} of a quantum channel $\mathcal{E}_{AB}$ is defined as
	\begin{equation}
		E(\mathcal{E}_{AB}) \coloneqq \min_{\mathcal{F}\in \clocc(A:B)} \max_{\mathcal{M}\in \cloccm(A\widetilde{A}:B\widetilde{B})} D_{m}(\mathcal{M}_{A\widetilde{A}B\widetilde{B}}\circ \mathcal{E}_{AB}\Vert \mathcal{M}_{A\widetilde{A}B\widetilde{B}}\circ \mathcal{F}_{AB}),
	\end{equation}
	where the ancillary systems $\widetilde{A}$ and $\widetilde{B}$ can be arbitrary.
\end{definition}
There is no preferred basis for entanglement, which is the cause of lack of simplified expression like Lemma~\ref{sup_thm: meas-coh gen power} for the \mcp{}.
\begin{theorem}
	The \mep{} of a quantum channel $\mathcal{E}_{A}$ is a resource monotone satisfying
	\begin{enumerate}
		\item Non-negativity and faithfulness: $E(\mathcal{E}_{AB}\ge 0)$. $E(\mathcal{E}_{AB}) = 0$ if and only if $\mathcal{E}_{AB}\in \clocc(A\!:\!B)$.
		\item Monotonicity: For a $\clocc$ channel $\mathcal{K}_{AB}$ and a unital $\clocc$ channel $\mathcal{L}_{AB}$, it follows that
		\begin{equation}
			E(\mathcal{K}_{AB}\circ \mathcal{E}_{AB}\circ \mathcal{L}_{AB}) \le E(\mathcal{E}_{AB}).
		\end{equation}
		\item Convexity: For quantum channels $\mathcal{E}_{AB}$ and $\mathcal{G}_{AB}$, it holds that $E(p \mathcal{E}_{AB}+(1-p) \mathcal{G}_{AB})\le pE(\mathcal{E}_{AB})+(1-p)E(\mathcal{G}_{AB})$ for $0\le p \le 1$.
	\end{enumerate}
\end{theorem}
\begin{proof}
	\begin{enumerate}
		\item Non-negativity and faithfulness can be shown as follows: $E(\mathcal{E}_{AB})\ge 0$ due to the non-negativity of $D_{m}$. $E(\mathcal{E}_{AB})=0$ for $\mathcal{E}_{AB}\in \clocc(A\!:\!B)$ from the definition. For $\mathcal{E}_{AB}\notin \clocc(A\!:\!B)$, consider a LOCC measurement $\mathcal{M}_{A\widetilde{A}B\widetilde{B}}$ that includes $\ketbra{\Phi^{+}}{\Phi^{+}}_{A\widetilde{A}B\widetilde{B}}$ as its POVM element, where $\ket{\Phi^{+}}_{A\widetilde{A}B\widetilde{B}}$ is a maximally entangled state. By the Choi-Jamiolkowski isomorphism, it holds that $\mathcal{E}_{AB}(\Phi^{+}_{A\widetilde{A}B\widetilde{B}})\neq \mathcal{F}_{AB}(\Phi^{+}_{A\widetilde{A}B\widetilde{B}})$ for any $\mathcal{F}_{AB}\in \clocc(A\!:\!B)$. Therefore, it holds that $E(\mathcal{E}_{AB}) = 0$ if and only if $\mathcal{E}_{AB}\in \clocc(A\!:\!B)$.
		\item For a $\clocc$ channel $\mathcal{K}_{AB}$ and a unital $\clocc$ channel $\mathcal{L}_{AB}$, it follows that
		\begin{align}
			E(\mathcal{K}_{AB}\circ \mathcal{E}_{AB}\circ \mathcal{L}_{AB})	& = \min_{\mathcal{F}\in \clocc(A:B)}\max_{\mathcal{M}\in \cloccm(A\widetilde{A}:B\widetilde{B})}\\
			&\qquad D_{m}(\mathcal{M}_{A\widetilde{A}B\widetilde{B}}\circ \mathcal{K}_{AB}\circ \mathcal{E}_{AB}\circ \mathcal{L}_{AB}\Vert \mathcal{M}_{A\widetilde{A}B\widetilde{B}}\circ \mathcal{F}_{AB})\\
			&\le \min_{\mathcal{F}\in \clocc(A:B)}\max_{\mathcal{M}\in \cloccm(A\widetilde{A}:B\widetilde{B})}\\
			&\qquad D_{m}(\mathcal{M}_{A\widetilde{A}B\widetilde{B}}\circ \mathcal{K}_{AB}\circ \mathcal{E}_{AB}\circ \mathcal{L}_{AB}\Vert \mathcal{M}_{A\widetilde{A}B\widetilde{B}}\circ \mathcal{K}_{AB}\circ \mathcal{F}_{AB}\circ \mathcal{L}_{AB})\\
			&\le \min_{\mathcal{F}\in \clocc(A:B)}\max_{\mathcal{M}\in \cloccm(A\widetilde{A}:B\widetilde{B})}\\
			&\qquad D_{m}(\mathcal{M}_{A\widetilde{A}B\widetilde{B}}\circ \mathcal{K}_{AB}\circ \mathcal{E}_{AB}\Vert \mathcal{M}_{A\widetilde{A}B\widetilde{B}}\circ \mathcal{K}_{AB}\circ \mathcal{F}_{AB})\\
			&\le \min_{\mathcal{F}\in \clocc(A:B)}\max_{\mathcal{M}\in \cloccm(A\widetilde{A}:B\widetilde{B})} \\ &\qquad D_{m}(\mathcal{M}_{A\widetilde{A}B\widetilde{B}}\circ \mathcal{E}_{AB}\Vert \mathcal{M}_{A\widetilde{A}B\widetilde{B}}\circ \mathcal{F}_{AB})\\
			&=E(\mathcal{E}_{AB}),
		\end{align}
		where the first inequality follows because $\mathcal{K}_{AB}\circ \mathcal{F}_{AB}\circ \mathcal{L}_{AB}\in \clocc(A\!:\!B)$, the second inequality holds from the monotonicity of $D_{m}$, and the third inequality is due to $\mathcal{M}_{A\widetilde{A}B\widetilde{B}}\circ \mathcal{K}_{AB}\in \cloccm(A\widetilde{A}\!:\! B\widetilde{B})$.
		\item Convexity:
		\begin{align}
			E(p \mathcal{E}_{AB}+(1-p) \mathcal{G}_{AB})& = \min_{\mathcal{F}\in \clocc(A:B)}\max_{\mathcal{M}\in \cloccm(A\widetilde{A}:B\widetilde{B})}\\
			&\qquad D_{m}(\mathcal{M}_{A\widetilde{A}B\widetilde{B}}\circ  \{p \mathcal{E}_{AB}+(1-p) \mathcal{G}_{AB}\}\Vert \mathcal{M}_{A\widetilde{A}B\widetilde{B}}\circ \mathcal{F}_{AB})\\
			& = \min_{\mathcal{F},\mathcal{F}'\in \clocc(A:B)}\max_{\mathcal{M}\in \cloccm(A\widetilde{A}:B\widetilde{B})}\\
			&\qquad D_{m}(\mathcal{M}_{A\widetilde{A}B\widetilde{B}}\circ  \{p \mathcal{E}_{AB}+(1-p) \mathcal{G}_{AB}\}\Vert\\
			&\qquad\qquad \mathcal{M}_{A\widetilde{A}B\widetilde{B}}\circ \{p\mathcal{F}_{AB}+(1-p)\mathcal{F'}_{AB})\\
			& \le \min_{\mathcal{F},\mathcal{F}'\in \clocc(A:B)}\max_{\mathcal{M}\in \cloccm(A\widetilde{A}:B\widetilde{B})}\\
			&\qquad \{pD_{m}(\mathcal{M}_{A\widetilde{A}B\widetilde{B}}\circ  \mathcal{E}_{AB}\Vert \mathcal{M}_{A\widetilde{A}B\widetilde{B}}\circ\mathcal{F}_{AB})\\
			&\qquad +(1-p)D_{m}(\mathcal{M}_{A\widetilde{A}B\widetilde{B}}\circ  \mathcal{G}_{AB}\Vert \mathcal{M}_{A\widetilde{A}B\widetilde{B}}\circ\mathcal{F'}_{AB})\}\\
			& \le \min_{\mathcal{F},\mathcal{F}'\in \clocc(A:B)}\\
			&\quad \bigg\{\max_{\mathcal{M}\in \cloccm(A\widetilde{A}:B\widetilde{B})}pD_{m}(\mathcal{M}_{A\widetilde{A}B\widetilde{B}}\circ  \mathcal{E}_{AB}\Vert \mathcal{M}_{A\widetilde{A}B\widetilde{B}}\circ\mathcal{F}_{AB})\\
			&\quad +\max_{\mathcal{M}'\in \cloccm(A\widetilde{A}:B\widetilde{B})}(1-p)D_{m}(\mathcal{M}_{A\widetilde{A}B\widetilde{B}}'\circ  \mathcal{G}_{AB}\Vert \mathcal{M}_{A\widetilde{A}B\widetilde{B}}'\circ\mathcal{F'}_{AB})\bigg\}\\
			&=pE(\mathcal{E}_{AB}) + (1-p) E(\mathcal{G}_{AB}),
		\end{align}
		where the joint convexity of $D_{m}$ is used for the first inequality.
	\end{enumerate}
\end{proof}

\section{Main results}
\subsection{Measurement-cohering power conversion to \mep{}}
Here we demonstrate our main results that the \mcp{} of a quantum channel can be converted to the \mep{} in composition with free channels. The following lemma will be used to prove the main results:
\begin{lemma}\label{sup_thm: no need for anc. meas.}
	For quantum channels $\mathcal{E}_{AB}$ and $\mathcal{F}_{AB}$, it holds that
	\begin{align}
		\max_{\mathcal{M}\in \cloccm(A\widetilde{A} : B\widetilde{B})}&D_{m}(\mathcal{M}_{A\widetilde{A}B\widetilde{B}}\circ \Delta_{AB}\circ\mathcal{E}_{AB}\Vert \mathcal{M}_{A\widetilde{A}B\widetilde{B}}\circ \Delta_{AB}\circ\mathcal{F}_{AB})\nonumber\\
		&=\max_{\mathcal{M}\in \cloccm(A : B)}D_{m}(\mathcal{M}_{AB}\circ \Delta_{AB}\circ\mathcal{E}_{AB}\Vert \mathcal{M}_{AB}\circ \Delta_{AB}\circ\mathcal{F}_{AB})
	\end{align}
\end{lemma}
\begin{proof}
	The left-hand side (LHS) is equal to or greater than the right-hand side (RHS) due to the broader domain being considered for maximization. The opposite direction is demonstrated below:
	\begin{align}
		&\max_{\mathcal{M}\in \cloccm(A\widetilde{A} : B\widetilde{B})}D_{m}(\mathcal{M}_{A\widetilde{A}B\widetilde{B}}\circ \Delta_{AB}\circ\mathcal{E}_{AB}\Vert 
		\mathcal{M}_{A\widetilde{A}B\widetilde{B}}\circ \Delta_{AB}\circ\mathcal{F}_{AB})\\
		&=\max_{\mathcal{M}\in \cloccm(A\widetilde{A} : B\widetilde{B})}\dfrac{1}{d^{4}}\sum_{x}D(\mathcal{E}_{AB}^{\dag}\circ\Delta_{AB}(M_{A\widetilde{A}B\widetilde{B}}^{x})\Vert  \mathcal{F}_{AB}^{\dag}\circ\Delta_{AB}(M_{A\widetilde{A}B\widetilde{B}}^{x}))\\
		&=\max_{\mathcal{M}\in \cloccm(A\widetilde{A} : B\widetilde{B})}\dfrac{1}{d^{4}}\sum_{x}D\bigg (\sum_{i,j}\mathcal{E}_{AB}^{\dag}(\ketbra{ij}{ij}_{AB})\otimes \bra{ij}M_{A\widetilde{A}B\widetilde{B}}^{x}\ket{ij}_{AB}\bigg\Vert \\
		&\qquad\qquad\qquad\qquad\quad \sum_{i,j} \mathcal{F}_{AB}^{\dag}(\ketbra{ij}{ij}_{AB})\otimes \bra{ij}M_{A\widetilde{A}B\widetilde{B}}^{x}\ket{ij}_{AB}\bigg )\\
		&\le \max_{\mathcal{M}\in \cloccm(A\widetilde{A} : B\widetilde{B})}\dfrac{1}{d^{4}}\sum_{x}\sum_{i,j} \tr_{\widetilde{A}\widetilde{B}}( \bra{ij}M_{A\widetilde{A}B\widetilde{B}}^{x}\ket{ij}_{AB})\\
		&\qquad\qquad\qquad\qquad\quad  D (\mathcal{E}_{AB}^{\dag}(\ketbra{ij}{ij}_{AB})\Vert \mathcal{F}_{AB}^{\dag}(\ketbra{ij}{ij}_{AB}) )\\
		&\le \dfrac{1}{d^{2}} \sum_{i,j} D (\mathcal{E}_{AB}^{\dag}(\ketbra{ij}{ij}_{AB})\Vert \mathcal{F}_{AB}^{\dag}(\ketbra{ij}{ij}_{AB}) )\\
		&=D_{m}(\mathcal{I}_{AB}\circ \mathcal{E}_{AB} \Vert \mathcal{I}_{AB}\circ \mathcal{F}_{AB})\\
		&=D_{m}(\mathcal{I}_{AB}\circ \Delta_{AB} \circ \mathcal{E}_{AB} \Vert \mathcal{I}_{AB}\circ \Delta_{AB} \circ \mathcal{F}_{AB})\\
		&\le \max_{\mathcal{M}\in \cloccm(A: B)} D_{m}(\mathcal{M}_{AB}\circ \Delta_{AB}\circ \mathcal{E}_{AB}\Vert \mathcal{M}_{AB}\circ \Delta_{AB}\circ \mathcal{F}_{AB}),
	\end{align}
	where the first inequality follows from the two facts~\cite{watrous2018TTQI}:
	\begin{equation}
		D(P_{0}+P_{1}\Vert Q_{0}+Q_{1}) \le D(P_{0}\Vert Q_{0}) + D(P_{1}\Vert Q_{1})
	\end{equation}
	for positive semidefinite matrices $P_{0},P_{1},Q_{0}$, and $Q_{1}$, and,
	\begin{equation}
		D(X\otimes Y\Vert K\otimes L) = \tr(Y) D(X\Vert K) + \tr(X) D(Y\Vert L)
	\end{equation}
	for positive semidefinite matrices $X\neq 0, Y\neq 0, K$ and $L$.
\end{proof}
\begin{theorem}\label{sup_thm: mcgp bounds megp}
	For a quantum channel $\mathcal{E}_{A}$, a $\dio$ channel $\mathcal{K}_{AB}$, and a unital $\dio$ channel $\mathcal{L}_{AB}$, it holds that
	\begin{equation}
		C(\mathcal{E}_{A}) \ge E(\Delta_{AB}\circ \mathcal{K}_{AB}\circ\mathcal{E}_{A}\circ \mathcal{L}_{AB}).
	\end{equation}
\end{theorem}
\begin{proof}
	Let $\mathcal{K}_{AB}$ be a $\dio$ channel and $\mathcal{L}_{AB}$ be a unital $\dio$ channel. We have that
	\begin{align}
		C(\mathcal{E}_{A})& \ge\min_{\mathcal{F}_{A}\in \dio} \max_{\mathcal{M}_{AB}\in \imeas} D_{m}\left ( \mathcal{M}_{AB}\circ \mathcal{E}_{A}\Vert  \mathcal{M}_{AB}\circ \mathcal{F}_{A}\right ) \\
		&\ge \min_{\mathcal{F}_{A}\in \dio} \max_{\mathcal{M}_{AB}\in \imeas} D_{m}\left ( \mathcal{M}_{AB}\circ \mathcal{K}_{AB}\circ\mathcal{E}_{A}\Vert  \mathcal{M}_{AB}\circ \mathcal{K}_{AB}\circ \mathcal{F}_{A}\right ) \\
		&\ge \min_{\mathcal{F}_{A}\in \dio} \max_{\mathcal{M}_{AB}\in \imeas} D_{m}\left ( \mathcal{M}_{AB}\circ \mathcal{K}_{AB}\circ\mathcal{E}_{A}\circ \mathcal{L}_{AB}\Vert  \mathcal{M}_{AB}\circ \mathcal{K}_{AB}\circ \mathcal{F}_{A}\circ \mathcal{L}_{AB}\right ) \\
		&= \min_{\mathcal{F}_{A}\in \dio} \max_{\mathcal{M}_{AB}\in \imeas} \nonumber \\
        &\qquad D_{m}\left ( \mathcal{M}_{AB}\circ \Delta_{AB}\circ \mathcal{K}_{AB}\circ\mathcal{E}_{A}\circ \mathcal{L}_{AB}\Vert  \mathcal{M}_{AB}\circ \Delta_{AB}\circ \mathcal{K}_{AB}\circ \mathcal{F}_{A}\circ \mathcal{L}_{AB}\right ) \\
		&= \min_{\mathcal{F}_{A}\in \dio} \max_{\mathcal{M}\in \cloccm (A:B)}\nonumber \\
        &\qquad D_{m}\left ( \mathcal{M}_{AB}\circ \Delta_{AB}\circ \mathcal{K}_{AB}\circ\mathcal{E}_{A}\circ \mathcal{L}_{AB}\Vert  \mathcal{M}_{AB}\circ \Delta_{AB}\circ \mathcal{K}_{AB}\circ \mathcal{F}_{A}\circ \mathcal{L}_{AB}\right ) \\
		&= \min_{\mathcal{F}_{A}\in \dio} \max_{\mathcal{M}\in \cloccm (A\widetilde{A}:B\widetilde{B})} \nonumber\\
        &\qquad D_{m}\left ( \mathcal{M}_{A\widetilde{A}B\widetilde{B}}\circ \Delta_{AB}\circ \mathcal{K}_{AB}\circ\mathcal{E}_{A}\circ \mathcal{L}_{AB}\Vert  \mathcal{M}_{A\widetilde{A}B\widetilde{B}}\circ \Delta_{AB}\circ \mathcal{K}_{AB}\circ \mathcal{F}_{A}\circ \mathcal{L}_{AB}\right )\\
		&\ge \min_{\mathcal{F}\in \clocc(A:B)} \max_{\mathcal{M}\in \cloccm (A\widetilde{A}:B\widetilde{B})} \nonumber\\
        &\qquad D_{m}\left ( \mathcal{M}_{A\widetilde{A}B\widetilde{B}}\circ \Delta_{AB}\circ \mathcal{K}_{AB}\circ\mathcal{E}_{A}\circ \mathcal{L}_{AB}\Vert  \mathcal{M}_{A\widetilde{A}B\widetilde{B}}\circ \Delta_{AB}\circ \mathcal{K}_{AB}\circ \mathcal{F}_{AB}\circ \mathcal{L}_{AB}\right )\\
		&=E(\Delta_{AB}\circ \mathcal{K}_{AB}\circ\mathcal{E}_{A}\circ \mathcal{L}_{AB}),
	\end{align}
	where the second inequality holds for $\mathcal{M}_{AB}\circ \mathcal{K}_{AB}\in\imeas$, the third inequality follows from the data-processing inequality of $D_{m}$. The third equality holds due to Lemma~\ref{sup_thm: no need for anc. meas.}.
\end{proof}
\begin{theorem}\label{sup_thm: mcgp to megp by cnot}
	For a quantum channel $\mathcal{E}_{A}$, it holds that
	\begin{equation}
		C(\mathcal{E}_{A}) \le E(\Delta_{AB}\circ \mathcal{E}_{A} \circ \mathcal{U}^{\dag}_{\cnot}),
	\end{equation}
	where the generalized CNOT channel $\mathcal{U}_{\cnot}$ is a unitary channel given by the unitary operator
	\begin{equation}
		U_{\cnot} = \sum_{i,j=0}^{d-1}\ketbra{i,i\oplus j}{ij}_{AB}
	\end{equation}
	with $\oplus$ denoting addition modulo $d$.
\end{theorem}
The proof of Theorem~\ref{sup_thm: mcgp to megp by cnot} needs a bound for the relative entropy of entanglement of a positive semidefinite operator as given below:
\begin{lemma}\label{sup_thm: rel. entropy bound}
	For a positive semidefinite matrix $X_{AB}\ge 0$, the relative entropy of entanglement of $X_{AB}$ is defined as
	\begin{equation}
		E_{R}(X_{AB}) = \min_{Y\in \sepo(A:B), \tr{Y}=\tr{X}} D(X_{AB}\Vert Y_{AB}),
	\end{equation}
	where $\sepo(A:B)$ is the set of separable operators~\cite{watrous2018TTQI}. The relative entropy of entanglement of $X_{AB}$ is lower-bounded as follows:
	\begin{equation}
		E_{R}(X_{AB})\ge \max\{ S(X_{A})-S(X_{AB}), S(X_{B})-S(X_{AB})\}.
	\end{equation}
\end{lemma}
For the proof of the above lemma, see Lemma 12 of Ref.~\cite{hjkim2022RelationQuantumCoherence}. The proof of Theorem~\ref{sup_thm: mcgp to megp by cnot} is stated below:
\begin{proof}
	\begin{align}
		E(\Delta_{AB}\circ \mathcal{E}_{A} \circ \mathcal{U}^{\cnot\dag}_{AB})& = \min_{\mathcal{F}\in \clocc(A:B)} \max_{\mathcal{M}\in \cloccm(A\widetilde{A}:B\widetilde{B})}\\
		&\qquad D_{m}(\mathcal{M}_{A\widetilde{A}B\widetilde{B}}\circ \Delta_{AB}\circ \mathcal{E}_{A}\circ \mathcal{U}^{\cnot \dag}_{AB}\Vert \mathcal{M}_{A\widetilde{A}B\widetilde{B}}\circ  \mathcal{F}_{AB}) \\
		& \ge \min_{\mathcal{F}\in \clocc(A:B)} \max_{\mathcal{M}\in \cloccm(A\widetilde{A}:B\widetilde{B})}\\
		&\qquad D_{m}(\mathcal{M}_{A\widetilde{A}B\widetilde{B}}\circ \Delta_{AB}\circ \mathcal{E}_{A}\circ \mathcal{U}^{\cnot \dag}_{AB}\Vert \mathcal{M}_{A\widetilde{A}B\widetilde{B}}\circ \Delta_{AB}\circ \mathcal{F}_{AB}) \\
		& = \min_{\mathcal{F}\in \clocc(A:B)} \max_{\mathcal{M}_{AB}\in \cloccm(A:B)}\\
		&\qquad D_{m}(\mathcal{M}_{AB}\circ \Delta_{AB}\circ \mathcal{E}_{A}\circ \mathcal{U}^{\cnot \dag}_{AB}\Vert \mathcal{M}_{AB}\circ \Delta_{AB}\circ  \mathcal{F}_{AB}) \\
		& = \min_{\mathcal{F}\in \clocc(A:B)} \max_{\mathcal{M}_{AB}\in \imeas} D_{m}(\mathcal{M}_{AB}\circ  \mathcal{E}_{A}\circ \mathcal{U}^{\cnot \dag}_{AB}\Vert \mathcal{M}_{AB}\circ  \mathcal{F}_{AB}) \\
		& = \min_{\mathcal{F}\in \clocc(A:B)} D_{m}(\mathcal{I}_{AB}\circ  \mathcal{E}_{A}\circ \mathcal{U}^{\cnot \dag}_{AB}\Vert \mathcal{I}_{AB}\circ  \mathcal{F}_{AB}) \\
		& = \min_{\mathcal{M}\in \cloccm(A:B)} D_{m}(\mathcal{I}_{AB}\circ  \mathcal{E}_{A}\circ \mathcal{U}^{\cnot \dag}_{AB}\Vert \mathcal{M}_{AB}) \\
		& \ge \dfrac{1}{d^{2}}\sum_{i,j} E_{R}(\mathcal{U}_{AB}^{\cnot}\circ \mathcal{E}_{A}^{\dag}(\ketbra{ij}{ij}_{AB}))\\
		&= \dfrac{1}{d}\sum_{i} E_{R}(\mathcal{U}_{AB}^{\cnot}\circ \mathcal{E}_{A}^{\dag}(\ketbra{i0}{i0}_{AB}))\\
		&\ge \dfrac{1}{d}\sum_{i} \{S(\Delta_{A}\circ \mathcal{E}_{A}^{\dag}(\ketbra{i}{i}_{A}))-S(\mathcal{E}_{A}^{\dag}(\ketbra{i}{i}_{A}))\}\\
		&= C(\mathcal{E}_{A}),
	\end{align}
	where in the first inequality we maximize over $\mathcal{M}_{A\widetilde{A}B\widetilde{B}}\circ \Delta_{AB}$ instead of $\mathcal{M}_{A\widetilde{A}B\widetilde{B}}$ reducing the domain of the maximization. The second equality follows from Lemma~\ref{sup_thm: no need for anc. meas.}. The fourth equality holds for Corollary~\ref{sup_thm: structure of a bipartite IM} and the monotonicity of the measurement relative entropy for classical post-processing channels. The fifth equality holds since any measurement in $\cloccm$ having $d^2$ outcomes can be implemented by a $\clocc$ channel followed by the incoherent basis measurement. The final inequality follows from Lemma~\ref{sup_thm: rel. entropy bound}.
\end{proof}

Theorem~\ref{sup_thm: mcgp bounds megp} and Theorem~\ref{sup_thm: mcgp to megp by cnot} combine to the main result:
\begin{theorem}\label{sup_thm: coh. conversion to ent.}
	For a quantum channel $\mathcal{E}_{A}$, it holds that
	\begin{equation}
		C(\mathcal{E}_{A}) = \max_{\mathcal{\mathcal{K_{AB},\mathcal{L}_{AB}}}} E(\Delta_{AB}\circ \mathcal{K}_{AB}\circ\mathcal{E}_{A} \circ \mathcal{L}_{AB}),
	\end{equation}
	where the maximum is taken over a detection-incoherent operation $\mathcal{K}_{AB}$ and a unital detection-incoherent operation $\mathcal{L}_{AB}$. The maximum is achieved as follows:
	\begin{equation}
		C(\mathcal{E}_{A}) = E(\Delta_{AB}\circ \mathcal{E}_{A} \circ \mathcal{U}^{\dag}_{\cnot}).
	\end{equation}
\end{theorem}
\begin{figure*}[th]
	\centering
	\begin{subfigure}[b]{0.45\textwidth}
		\centering
		\includegraphics[width=\textwidth]{"fig_mcgp_to_megp_conversion"}
		\caption{}
		\label{fig-A: mcgp to megp conversion}
	\end{subfigure}
	\hspace{1cm}
	\begin{subfigure}[b]{0.45\textwidth}
		\centering
		\includegraphics[width=\textwidth]{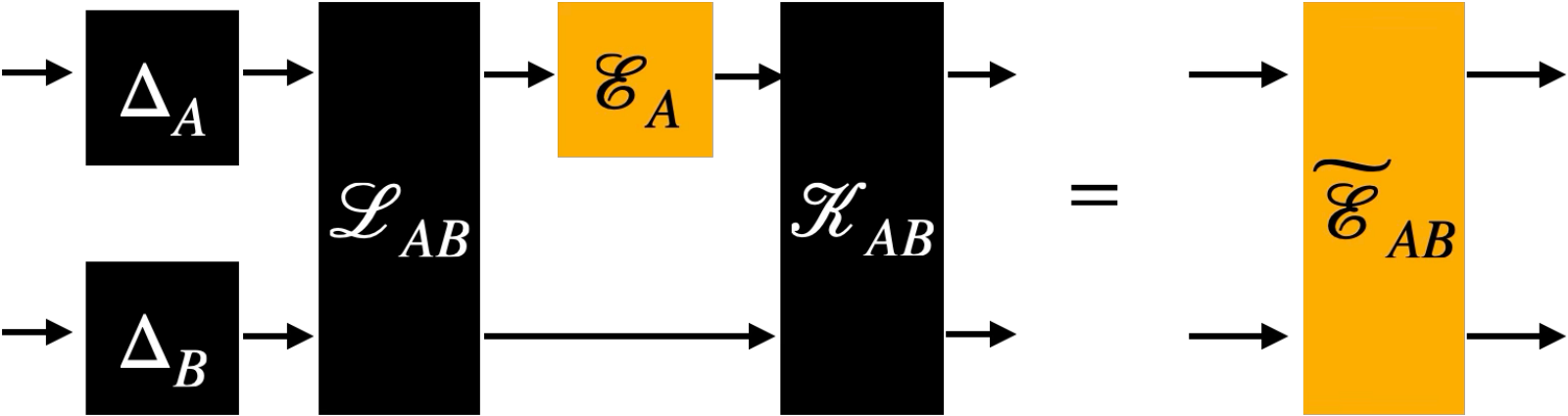}
		\caption{}
		\label{fig-A: scgp to segp conversion}
	\end{subfigure}
	
	\caption{(a) Measurement-cohering power of a quantum channel $\mathcal{E}_{A}$ can convert to \mep{} of a quantum channel $\widetilde{\mathcal{E}}_{AB}$ under a $\dio$ channel $\mathcal{K}_{AB}$, a unital $\dio$ channel $\mathcal{L}_{AB}$, and the dephasing channel $\Delta_{AB}$. (b) State-cohering power of a quantum channel $\mathcal{E}_{A}$ can convert to \sep{} of a quantum channel $\widetilde{\mathcal{E}}_{AB}$ under $\mio$ channels $\mathcal{K}_{AB}$ and $\mathcal{L}_{AB}$, and the dephasing channel $\Delta_{AB}$~\cite{theurer2020QuantifyingDynamicalCoherence}.}
\end{figure*}
Note that $\mathcal{U}_{\cnot}^{\dag}$ is in $\dio$ since  $\mathcal{U}_{\cnot}$ is in $\mio$. This result says that the \mcp{} of a quantum channel $\mathcal{E}_{A}$ can be converted into the \mep{} of a bipartite channel without any additional \mcp{}s. The post-dephasing channels on the right-hand side act on a quantum measurement, making it incoherent with respect to the incoherent basis, as depicted in Fig.~\ref{fig-A: mcgp to megp conversion}. This allows the adjoint channel of the generalized CNOT channel $\mathcal{U}_{\cnot}^{\dag}$ to effectively entangle each POVM element.  Conversely, when converting \scp{} to \sep{}, dephasing channels function as a pre-processing step to ensure the input states are incoherent, as illustrated in Fig.~\ref{fig-A: scgp to segp conversion}~\cite{theurer2020QuantifyingDynamicalCoherence}.

Next, we show that the \mep{} induces the \mcp{}.
\begin{theorem}\label{sup_thm: megp induces mcgp}
	The \mep{} monotone $\widetilde{E}$ induces the \mcp{} monotone $\widetilde{C}$ given by
	\begin{equation}
		\widetilde{C}(\mathcal{E}_{A}) = \max_{\mathcal{\mathcal{K_{AB},\mathcal{L}_{AB}}}} \widetilde{E}(\Delta_{AB}\circ \mathcal{K}_{AB}\circ \mathcal{E}_{A}\circ \mathcal{L}_{AB}),
	\end{equation}
	where the maximum is taken over a detection-incoherent operation $\mathcal{K}_{AB}$ and a unital detection-incoherent operation $\mathcal{L}_{AB}$.
\end{theorem}
The proof of Theorem~\ref{sup_thm: megp induces mcgp} needs the following fact:
\begin{proposition}\label{sup_thm: incoh. meas. sandwiched ch. is locc}
	For a quantum channel $\mathcal{E}_{A}$, a detection-incoherent operation $\mathcal{K}_{AB}$, and a unital detection-incoherent operation $\mathcal{L}_{AB}$, it holds that $\Delta_{AB}\circ \mathcal{K}_{AB}\circ \mathcal{E}_{A}\circ\mathcal{L}_{AB}\circ \Delta_{AB}\in \clocc(A\!:\!B)$.
\end{proposition}
\begin{proof}
	For an arbitrary input state $\rho_{AB}$, the composite channel results in
	\begin{align}
		\Delta_{AB}\circ \mathcal{K}_{AB}\circ \mathcal{E}_{A}\circ\mathcal{L}_{AB}\circ \Delta_{AB}(\rho_{AB}) & =\Delta_{AB}\circ \mathcal{K}_{AB}\circ \mathcal{E}_{A}\circ\mathcal{L}_{AB}\left (\sum_{i,j} p_{ij}\ketbra{ij}{ij}_{AB}\right ) \\
		 &=\sum_{i,j}p_{ij}\sum_{k,l} q_{kl}^{(ij)}\ketbra{kl}{kl}_{AB}\\
		 &=\sum_{i,j}p_{ij} \sigma^{(ij)},
	\end{align}
	which is a statistical mixture of separable states $\sigma^{(ij)}$. This channel can be implemented by first measuring the input state in the computational basis, then preparing the states $\sigma^{(ij)}$ according to the outcome. All of the procedure can be done in LOCC.
\end{proof}

Here's the proof of Theorem~\ref{sup_thm: megp induces mcgp}:
\begin{proof}
	\begin{enumerate}
		\item Non-negativity and faithfulness can be shown as follows: $\widetilde{C}(\mathcal{E}_{A})\ge 0$ due to the non-negativity of the entanglement monotone $\widetilde{E}$. If $\mathcal{E}_{A}\in \dio$, then
		\begin{equation}
			\Delta_{AB}\circ \mathcal{K}_{AB}\circ \mathcal{E}_{A}\circ \mathcal{L}_{AB}=\Delta_{AB}\circ \mathcal{K}_{AB}\circ \mathcal{E}_{A}\circ \mathcal{L}_{AB}\circ \Delta_{AB}
		\end{equation}
		is a LOCC channel by Proposition~\ref{sup_thm: incoh. meas. sandwiched ch. is locc}. Therefore $\widetilde{C}(\mathcal{E_{A}})=0$. If $\mathcal{E}_{A}\notin \dio$, then Theorem~\ref{sup_thm: coh. conversion to ent.} implies that $\Delta_{AB}\circ \mathcal{E}_{A}\circ \mathcal{U}_{\cnot}\notin \clocc(A\!:\!B)$, thereby $\widetilde{C}(\mathcal{E}_{A})> 0$.
		\item Convexity holds due to convexity of the entanglement monotone $ \widetilde{E}$.
		\item For a detection-incoherent channel $\mathcal{K}'_{A}$ and a unital detector-incoherent channel $\mathcal{L}'_{A}$, the monotonicity holds as
		\begin{align}
			\widetilde{C}(\mathcal{K}'_{A}\circ \mathcal{E}_{A}\circ \mathcal{L}'_{A}) &= \max_{\mathcal{\mathcal{K_{AB},\mathcal{L}_{AB}}}} E(\Delta_{AB}\circ \mathcal{K}_{AB}\circ (\mathcal{K}'_{A}\circ \mathcal{E}_{A}\circ \mathcal{L}'_{A})\circ \mathcal{L}_{AB}) \\
			& \le \widetilde{C}(\mathcal{E}_{A}),
		\end{align}
		
		because $\mathcal{K}'_{A}$ and $\mathcal{L}'_{A}$ reduce the domain of maximization of $\mathcal{K}_{AB}$ and $\mathcal{L}_{AB}$, respectively.
	\end{enumerate}
\end{proof}

\subsection{Reduction of measurement-resource powers to measurement-resources}
When we consider a measurement channel $\mathcal{N}_{A}$ with output system $A'$, Theorem~\ref{sup_thm: coh. conversion to ent.} reduces to the measurement-coherence conversion to measurement-entanglement~\cite{hjkim2022RelationQuantumCoherence}. It holds because $\mathcal{N}_{A}=\Delta_{A'}\circ \mathcal{N}_{A} = \mathcal{I}_{A'}\circ \mathcal{N}_{A}$.

The \mcp{} of a measurement channel $\mathcal{N}_{A}$ reduces to the measurement-coherence of the measurement channel $\mathcal{N}_{A}$ by Lemma~\ref{sup_thm: meas-coh gen power} as follows:
\begin{align}
	C(\mathcal{N}_{A}) &= C_{m} (\mathcal{I}_{A'}\circ \mathcal{N}_{A})\\
	&=C_{m} (\mathcal{N}_{A}).
\end{align}

Similarly, the \mep{} $E(\mathcal{N}_{AB})$ of a measurement channel $\mathcal{N}_{AB}$ reduces to the measurement-entanglement $E_{m}(\mathcal{N}_{AB})$ of a measurement channel $\mathcal{N}_{AB}$. We first prove that $E(\mathcal{N}_{AB})\le E_{m}(\mathcal{N}_{AB})$:
\begin{align}
	E(\mathcal{N}_{AB}) &\coloneqq \min_{\mathcal{F}\in \clocc(A:B)} \max_{\mathcal{M}\in \cloccm(A\widetilde{A}:B\widetilde{B})}\\
	&\qquad D_{m}(\mathcal{M}_{A\widetilde{A}B\widetilde{A}}\circ \mathcal{N}_{AB}\Vert \mathcal{M}_{A\widetilde{A}B\widetilde{B}}\circ \mathcal{F}_{AB}) \\
	&\le \min_{\mathcal{F}\in \cloccm(A:B)} \max_{\mathcal{M}\in \cloccm(A\widetilde{A}:B\widetilde{B})}\\
	&\qquad D_{m}(\mathcal{M}_{A\widetilde{A}B\widetilde{A}}\circ \mathcal{N}_{AB}\Vert \mathcal{M}_{A\widetilde{A}B\widetilde{B}}\circ \mathcal{F}_{AB}) \\
	&=\min_{\mathcal{F}\in \cloccm(A:B)} \max_{\mathcal{M}\in \cloccm(A\widetilde{A}:B\widetilde{B})}\\
	&\qquad D_{m}(\mathcal{M}_{A\widetilde{A}B\widetilde{B}}\circ \Delta_{AB}\circ \mathcal{N}_{AB}\Vert \mathcal{M}_{A\widetilde{A}B\widetilde{B}}\circ \Delta_{AB}\circ \mathcal{F}_{AB}) \\
	& =\min_{\mathcal{F}\in \cloccm(A:B)} \max_{\mathcal{M}\in \cloccm(A:B)}\\
	&\qquad D_{m}(\mathcal{M}_{AB}\circ \Delta_{AB}\circ \mathcal{N}_{AB}\Vert \mathcal{M}_{AB}\circ \Delta_{AB}\circ\mathcal{F}_{AB}) \\
	& =\min_{\mathcal{F}\in \cloccm(A:B)} \max_{\mathcal{S}_{AB}\in \mathbf{S}} \\
	&\qquad D_{m}(\mathcal{S}_{AB}\circ \Delta_{AB}\circ \mathcal{N}_{AB}\Vert\mathcal{S}_{AB}\circ \Delta_{AB}\circ\mathcal{F}_{AB}) \\
	& =\min_{\mathcal{F}\in \cloccm(A:B)} D_{m}( \mathcal{N}_{AB}\Vert\mathcal{F}_{AB}) \\
	&= E_{m}(\mathcal{N}_{AB}),
\end{align}
where the third equality follows from Lemma~\ref{sup_thm: no need for anc. meas.} and the fact that measurement channels output incoherent states. The fourth equality holds from Corollary~\ref{sup_thm: structure of a bipartite IM}. The fifth equality follows from the monotonicity of $D_{m}$ under classical post-processing channel.

The opposite direction is proved as follows:
\begin{align}
	E(\mathcal{N}_{AB}) &\coloneqq \min_{\mathcal{F}\in \clocc(A:B)} \max_{\mathcal{M}\in \cloccm(A\widetilde{A}:B\widetilde{B})}\\
	&\qquad D_{m}(\mathcal{M}_{A\widetilde{A}B\widetilde{A}}\circ \mathcal{N}_{AB}\Vert \mathcal{M}_{A\widetilde{A}B\widetilde{B}}\circ \mathcal{F}_{AB}) \\
	&\ge \min_{\mathcal{F}\in \clocc(A:B)} \max_{\mathcal{M}\in \cloccm(A\widetilde{A}:B\widetilde{B})}\\
	&\qquad D_{m}(\mathcal{M}_{A\widetilde{A}B\widetilde{B}}\circ \Delta_{AB}\circ \mathcal{N}_{AB}\Vert \mathcal{M}_{A\widetilde{A}B\widetilde{B}}\circ \Delta_{AB}\circ \mathcal{F}_{AB}) \\
	& =\min_{\mathcal{F}\in \clocc(A:B)} \max_{\mathcal{M}\in \cloccm(A:B)}\\
	&\qquad D_{m}(\mathcal{M}_{AB}\circ \Delta_{AB}\circ \mathcal{N}_{AB}\Vert \mathcal{M}_{AB}\circ \Delta_{AB}\circ\mathcal{F}_{AB}) \\
	& =\min_{\mathcal{F}\in \cloccm(A:B)} \max_{\mathcal{M}\in \cloccm(A:B)}\\
	&\qquad D_{m}(\mathcal{M}_{AB}\circ \Delta_{AB}\circ \mathcal{N}_{AB}\Vert \mathcal{M}_{AB}\circ \Delta_{AB}\circ\mathcal{F}_{AB}) \\
	&= E_{m}(\mathcal{N}_{AB}),
\end{align}
where the first inequality holds due to $\mathcal{M}_{A\widetilde{A}B\widetilde{B}}\circ \Delta_{AB}\in \cloccm(A\widetilde{A}\!:\!B\widetilde{B})$, the second equality follows from Lemma~\ref{sup_thm: no need for anc. meas.}. The third equality holds because of $\Delta_{AB}\circ\mathcal{F}_{AB}\in \cloccm(A\!:\!B)$.

\subsection{Relation between \scp{} and \mcp{}}
Here we investigate the relation between a quantum channel's state-resource and measurement-resource powers. We begin by defining the \scp{} of a quantum channel $\mathcal{E}_{A}$ as
\begin{equation}
	C_{g}(\mathcal{E}_{A})\coloneqq \max_{\sigma_{A\widetilde{A}}\in \textbf{I}} C_{R}(\mathcal{E}_{A}(\sigma_{A\widetilde{A}})),
\end{equation}
where $C_{R}(\rho_{A})\coloneqq \min_{\sigma_{A}\in \mathbf{I}} D(\rho_{A}\vert \sigma_{A})$ is the relative entropy of coherence~\cite{baumgratz2014QuantifyingCoherence}, and $\mathbf{I}$ is the set of incoherent states. The \scp{} of a quantum channel $\mathcal{E}_{A}$ can be written as
\begin{equation}\label{sup_eq: coh. gen. power over incoh. basis}
	C_{g}(\mathcal{E}_{A}) = \max_{\ketbra{i}{i}_{A}\in \textbf{I}} C_{R}(\mathcal{E}_{A}(\ketbra{i}{i}_{A})).
\end{equation}
This can be shown as follows:
\begin{align}
	C_{g}(\mathcal{E}_{A})&\coloneqq \max_{\sigma_{A\widetilde{A}}\in \textbf{I}} C_{R}(\mathcal{E}_{A}(\sigma_{A\widetilde{A}}))\\
	&=C_{R}(\mathcal{E}_{A}(\sigma_{A\widetilde{A}}^{\ast}))\\
	&=C_{R}\left (\mathcal{E}_{A}\left (\sum_{i,j}p_{ij}\ketbra{ij}{ij}_{A\widetilde{A}}\right )\right )\\
	&\le \sum_{i,j}p_{ij} C_{R}(\mathcal{E}_{A}(\ketbra{ij}{ij}_{A\widetilde{A}}))\\
	&\le \max_{\ketbra{i}{i}_{A}\in \textbf{I}} C_{R}(\mathcal{E}_{A}(\ketbra{i}{i}_{A})),
\end{align}
where $\sigma_{A\widetilde{A}}^\ast$ is an optimal argument in the definition of $C_{g}(\mathcal{E}_{A})$.
The other direction $C_{g}(\mathcal{E}_{A})\ge \max_{\ketbra{i}{i}_{A}\in \textbf{I}} C_{R}(\mathcal{E}_{A}(\ketbra{i}{i}_{A})) $ holds because the domain of the maximization in the definition of $C_{g}$ includes the incoherent basis.

In general, there is no relation between a quantum channel's \scp{} and \mcp{}. This can be seen by examining some quantum channels. Consider the following preparation channel:
\begin{equation}
	\mathcal{E}_{A}(\rho_{A}) = \bra{0}\rho_{A}\ket{0}_{A}\ketbra{+}{+}_{A} + \bra{1}\rho_{A}\ket{1}_{A}\ketbra{+}{+}_{A} = \ketbra{+}{+}_{A}.
\end{equation}
The quantum channel $\mathcal{E}_{A}$ has maximum \scp{} as
\begin{equation}
	C_{g}(\mathcal{E}_{A}) = \max_{\sigma_{A}\in \text{I}} C_{R}(\mathcal{E}_{A}(\rho)_{A}) = 1,
\end{equation}
while its \mcp{} is zero as
\begin{equation}
	C(\mathcal{E}_{A}) = C_{m}(\mathcal{I}_{A}\circ \mathcal{E}_{A}) = 0,
\end{equation}
because $\mathcal{I}_{A}\circ \mathcal{E}_{A} = \{\frac{I_{A}}{2},\frac{I_{A}}{2}\}$.

Similarly, one can conceive a quantum channel $\mathcal{G}_{A}$ such that
\begin{equation}
	\mathcal{G}_{A}(\rho_{A}) =\bra{+}\rho_{A}\ket{+}_{A}\ketbra{0}{0}_{A} + \bra{-}\rho_{A}\ket{-}_{A}\ketbra{1}{1}_{A}.
\end{equation}
The quantum channel $\mathcal{G}_{A}$ does not generate any state-coherence while it can generate measurement-coherence as
\begin{equation}
	\mathcal{I}_{A}\circ \mathcal{G}_{A} = \{ \ketbra{+}{+}_{A}, \ketbra{-}{-}_{A} \}.
\end{equation}

Thus, it is not possible to assert a relation between the \scp{} and \mcp{} of a quantum channel. However, it can be observed that a unital quantum channel is related to its adjoint channel in terms of static resource powers as follows:
\begin{theorem}
	For a unital quantum channel $\mathcal{E}_{A}$, the \scp{} of $\mathcal{E}_{A}$ and \mcp{} of its adjoint channel $\mathcal{E}_{A}^{\dag}$ satisfy the following relation:
	\begin{equation}
		\dfrac{1}{d}C_{g}(\mathcal{E}_{A}) \le C(\mathcal{E}_{A}^{\dag})\le C_{g}(\mathcal{E}_{A}).
	\end{equation}
\end{theorem}
\begin{proof}
	It holds that
	\begin{align}
		C(\mathcal{E}_{A}^{\dag}) & =C_{m}(\mathcal{I}_{A}\circ \mathcal{E}_{A}^{\dag}) \\
		& =\dfrac{1}{d}\sum_{i=0}^{d-1} C_{R}(\mathcal{E}_{A}(\ketbra{i}{i}_{A}))\\
		&\le C_{g}(\mathcal{E}_{A}),
	\end{align}
	where the inequality follows from Eq.~\eqref{sup_eq: coh. gen. power over incoh. basis}. The lower bound follows from Eq.~\eqref{sup_eq: coh. gen. power over incoh. basis} and the non-negativity of $C_{R}$.
\end{proof}
This shows the equivalence of the \scp{} of a unital channel and the \mcp{} of its adjoint channel in the sense that a unital quantum channel having no \scp{} can neither generate measurement-coherence, nor vice versa.

State-entangling power of a quantum channel $\mathcal{E}_{AB}$ can be defined as
\begin{equation}
	E_{g}(\mathcal{E}_{AB})=\max_{\sigma \in \sepd(A\widetilde{A}:B\widetilde{B})} E_{R}(\mathcal{E}_{AB}(\sigma_{A\widetilde{A}B\widetilde{B}})).
\end{equation}
The ancillary system in the definition is necessary to capture the \sep{} of a quantum channel in a situation where the quantum channel acts on subsystems of a large system. Typical example is the SWAP gate on system $AB$. It does not generate entanglement from $\sigma_{AB}\in\sepd (A\!:\!B)$ while it can generate the maximum entanglement from $\sigma_{A\widetilde{A}B\widetilde{B}}\in\sepd(A\widetilde{A}\!:\!B\widetilde{B})$~\cite{nielsen2003QuantumDynamicsPhysical}.

A quantum channel's \sep{} is equivalent to its adjoint channel's \mep{} in the following sense:
\begin{theorem}
	For a unital quantum channel $\mathcal{E}_{AB}$, the \sep{} of $\mathcal{E}_{AB}$ and the \mep{} of its adjoint channel $\mathcal{E}_{AB}^{\dag}$ satisfy the following relation:
	\begin{equation}
		\dfrac{1}{d^{4}}E_{g}(\mathcal{E}_{AB}) \le E(\mathcal{E}_{AB}^{\dag})\le E_{g}(\mathcal{E}_{AB}).
	\end{equation}
\end{theorem}
\begin{proof}
	We first show the upper bound on $E(\mathcal{E}^\dag)$. Let $ \overline{\mathbf{ULOCC}}(A\!:\!B)=\{\mathcal{F}_{AB}\in \clocc(A\!:\!B):\mathcal{F}_{AB}(I_{AB})=I_{AB}\}$ be the set of unital $\clocc$ channels. The following holds that:
	\begin{align}
		E(\mathcal{E}_{AB}^{\dag}) & =\min_{\mathcal{F}\in \clocc(A:B)}\max_{\mathcal{M}\in \cloccm(A\widetilde{A}:B\widetilde{B})} D_{m}(M_{A\widetilde{A}B\widetilde{B}}\circ \mathcal{E}_{AB}^{\dag}\Vert \mathcal{M}_{A\widetilde{A}B\widetilde{B}}\circ \mathcal{F}_{AB})\\
		& \le\min_{\mathcal{F}\in \overline{\mathbf{ULOCC}}(A:B)}\max_{\mathcal{M}\in \cloccm(A\widetilde{A}:B\widetilde{B})} D_{m}(M_{A\widetilde{A}B\widetilde{B}}\circ \mathcal{E}_{AB}^{\dag}\Vert \mathcal{M}_{A\widetilde{A}B\widetilde{B}}\circ \mathcal{F}_{AB})\\
		& =D_{m}(\mathcal{M}_{A\widetilde{A}B\widetilde{B}}^{\ast} \circ \mathcal{E}_{AB}^{\dag}\Vert \mathcal{M}_{A\widetilde{A}B\widetilde{B}}^{\ast}\circ \mathcal{F}_{AB}^{\ast})\\
		&=\dfrac{1}{d^{4}}\sum_{x,y}D(\mathcal{E}_{AB}(M_{xy}^{\ast})\Vert \mathcal{F}_{AB}^{\ast\dag}(M_{xy}^{\ast}))\\
		&\le \max_{x,y}D\left (\mathcal{E}_{AB}\left (\dfrac{M_{xy}^{\ast}}{d^{4}}\right )\bigg\Vert \mathcal{F}_{AB}^{\ast\dag}\left (\dfrac{M_{xy}^{\ast}}{d^{4}}\right )\right )\\
		&=D(\mathcal{E}_{AB}(M_{0}^{\ast})\Vert \mathcal{F}_{AB}^{\ast \dag}(M_{0}^{\ast})).
	\end{align}
	Let $M_{0}^{\ast}=q_{0}\sigma_{A\widetilde{A}B\widetilde{B}} $, where $q_{0}=\tr M_{0}^{\ast}\le 1$ and $\sigma_{A\widetilde{A}B\widetilde{B}}=M_{0}^{\ast}/q_{0}$ is a separable state in $\sepd(A\widetilde{A}\!:\!B\widetilde{B})$. Then it follows that
	\begin{align}
		E(\mathcal{E}_{AB}^{\dag})& \le D(\mathcal{E}_{AB}(q_{0}\sigma_{A\widetilde{A}B\widetilde{B}})\Vert \mathcal{F}_{AB}^{\ast \dag}(q_{0}\sigma_{A\widetilde{A}B\widetilde{B}})) \\
		&=q_{0}D(\mathcal{E}_{AB}(\sigma_{A\widetilde{A}B\widetilde{B}})\Vert \mathcal{F}_{AB}^{\ast \dag}(\sigma_{A\widetilde{A}B\widetilde{B}}))\\
		&\le E_{g}(\mathcal{E}_{AB}).
	\end{align}
	
	The lower bound on $ E(\mathcal{E}_{AB}^{\dag})$ can be derived as follows:
	\begin{align}
		E_{g}(\mathcal{E}_{AB})&=\max_{\sigma \in \sepd(A\widetilde{A}:B\widetilde{B})} E_{R}(\mathcal{E}_{AB}(\sigma_{A\widetilde{A}B\widetilde{B}}))\\
		&=\max_{\sigma \in \sepd(A\widetilde{A}:B\widetilde{B})} \min_{\tau\in \sepd(A\widetilde{A}:B\widetilde{B})} D(\mathcal{E}_{AB}(\sigma_{A\widetilde{A}B\widetilde{B}}) \Vert \tau_{A\widetilde{A}B\widetilde{B}})\\
		 &= D(\mathcal{E}_{AB}(\sigma_{A\widetilde{A}B\widetilde{B}}^{\ast})\Vert \tau_{A\widetilde{A}B\widetilde{B}}^{\ast}).
	\end{align}
	We construct a measurement $\mathcal{N}_{AB}$ in $\cloccm(A\widetilde{A}:B\widetilde{B})$ that includes $\sigma_{A\widetilde{A}B\widetilde{B}}$ as one of its POVM elements. Then we observe that
	\begin{align}
		E(\mathcal{E}_{AB}^{\dag}) & =\min_{\mathcal{F}\in \clocc(A:B)}\max_{\mathcal{M}\in \cloccm(A\widetilde{A}:B\widetilde{B})}  D_{m}(\mathcal{M}_{A\widetilde{A}B\widetilde{B}}\circ \mathcal{E}_{AB}^{\dag}\Vert \mathcal{M}_{A\widetilde{A}B\widetilde{B}}\circ \mathcal{F}_{AB})\\
		&=D_{m}(\mathcal{M}_{A\widetilde{A}B\widetilde{B}}^{\ast} \circ \mathcal{E}_{AB}^{\dag}\Vert \mathcal{M}_{A\widetilde{A}B\widetilde{B}}^{\ast}\circ \mathcal{F}_{AB}^{\ast})\\
		&\ge D_{m}(\mathcal{N}_{A\widetilde{A}B\widetilde{B}} \circ \mathcal{E}_{AB}^{\dag}\Vert \mathcal{N}_{A\widetilde{A}B\widetilde{B}}\circ \mathcal{F}_{AB}^{\ast})\\
		&=\dfrac{1}{d^{4}}\sum_{x,y}D(\mathcal{E}_{AB}(N_{xy})\Vert \mathcal{F}_{AB}^{\ast\dag}(N_{xy}))\\
		&\ge \dfrac{1}{d^{4}} D(\mathcal{E}_{AB}(\sigma_{A\widetilde{A}B\widetilde{B}}^{\ast}) \Vert \tau_{A\widetilde{A}B\widetilde{B}}^{\ast})\\
		&\ge \dfrac{1}{d^{4}} E_{g}(\mathcal{E}_{AB}).
	\end{align}
    This completes the proof.
\end{proof}

\end{document}